%% file: main.tex
\documentclass[aps,prx,reprint,superscriptaddress,longbibliography,nofootinbib,amsmath,amssymb]{revtex4-2}

\usepackage{graphicx}
\usepackage{booktabs}
\usepackage{amsfonts}
\usepackage{amsthm}
\usepackage{hyperref}
\usepackage{cleveref}
\crefname{figure}{Fig.}{Figs.}
\Crefname{figure}{Fig.}{Figs.}

\theoremstyle{plain}
\newtheorem{theorem}{Theorem}
\newtheorem{proposition}{Proposition}
\newtheorem{lemma}{Lemma}
\newtheorem{remark}{Remark}

\newcommand{\KZ}{\mathcal{K}_Z}
\newcommand{\FZ}{F_Z}
\newcommand{\JZ}{J_Z}
\newcommand{\rH}{r_H}
\newcommand{\Bstar}{B^{\star}}
\newcommand{\sigmin}{\sigma_{\min}}
\newcommand{\etav}{\boldsymbol{\eta}}

\providecommand{\PARstart}[2]{#1#2}

\makeatletter
\AtBeginDocument{\immediate\write\@auxout{\string\citation{apsrev42Control}}}
\makeatother

\begin{document}

\title{Identifying the Sign of Coherent Over-Rotations with Logarithmically Many Pauli Settings}

\author{Yi Pan}
\thanks{Y.P. and M.H.T. contributed equally to this work.}
\affiliation{School of Computing, University of Georgia, Athens, GA 30602, USA}
\affiliation{GyriQAI, Inc., Athens, GA 30602, USA}

\author{Meng Hsiu Tsai}
\thanks{Y.P. and M.H.T. contributed equally to this work.}
\affiliation{Department of Computer Science and Engineering, University of Tennessee at Chattanooga, Chattanooga, TN 37403, USA}

\author{Weihang You}
\affiliation{School of Computing, University of Georgia, Athens, GA 30602, USA}
\affiliation{GyriQAI, Inc., Athens, GA 30602, USA}

\author{Hanqi Jiang}
\affiliation{School of Computing, University of Georgia, Athens, GA 30602, USA}
\affiliation{GyriQAI, Inc., Athens, GA 30602, USA}

\author{Junhao Chen}
\affiliation{School of Computing, University of Georgia, Athens, GA 30602, USA}
\affiliation{GyriQAI, Inc., Athens, GA 30602, USA}

\author{Wei Zhang}
\affiliation{School of Computer and Cyber Sciences, Augusta University, Augusta, GA 30912, USA}

\author{Isaac Lyngaas}
\affiliation{National Center for Computational Sciences, Oak Ridge National Laboratory, Oak Ridge, TN 37830, USA}

\author{Yingfeng Wang}
\email[Contact author: ]{yingfeng-wang@utc.edu}
\affiliation{Department of Computer Science and Engineering, University of Tennessee at Chattanooga, Chattanooga, TN 37403, USA}

\author{Tianming Liu}
\email[Contact author: ]{tliu@uga.edu}
\affiliation{School of Computing, University of Georgia, Athens, GA 30602, USA}
\affiliation{GyriQAI, Inc., Athens, GA 30602, USA}

\begin{abstract}
Calibrating a quantum processor means estimating gate-error parameters from data, and a hard-to-estimate
parameter is usually assumed to leave a weak signature more repetitions will resolve. Coherent
over-rotations break that premise. For commuting single- and two-qubit transverse over-rotations with
known support on a computational-basis input, the passive histogram is exactly invariant
under a sign group acting on the coherent angles, of order $2^{n+1}$ for the complete family with $n\ge2$, so
no estimator resolves the signs uniformly at any sample size. A calibration correction applied with the wrong sign doubles the rotation error it should remove. At zero angle the
obstruction turns continuous, with a Fisher information singular
along every coherent direction and an infinite Cram\'er--Rao bound. At generic angles the continuous defect clears for $5\le n\le8$, the range we compute, while the sign degeneracy persists. Covering the support removes both, and a design resolves the signed angles uniformly exactly when it covers. On $|\eta_R|<\pi/2$, $0\le p_R<1/2$ coverage is constructive, returning each angle in closed form from a ratio of two measured expectations without knowing the rates. A twirl-free code of
$\lceil\log_2(n+1)\rceil$ added product-Pauli settings covers, and for the complete family none uses fewer. Conditioning then sets the Cram\'er--Rao cost, and we give its floor explicitly. Exact computation matches the theory, the passive fit stays at
chance on the signs at every budget while the closed-form inverse recovers them, and angle magnitudes on two IBM Heron processors follow the conditioning ordering.
\end{abstract}

\maketitle

\input{sections/0_intro}
\input{sections/1_setup}
\input{sections/2_veil}
\input{sections/3_design}
\input{sections/4_input}
\input{sections/5_experiments}
\input{sections/6_related}
\input{sections/7_conclusion}

\input{sections/statements}

\bibliographystyle{apsrev4-2}
\bibliography{references,revtexcontrol}

\appendix
\crefalias{section}{appendix}
\input{appendix}

\end{document}

%% file: sections/0_intro.tex
\section{Introduction}

\PARstart{E}{very} quantum computer must be calibrated. Characterizing its gate errors is an estimation problem, in which physical error parameters are recovered from measurement data and the Fisher information of that data bounds every locally unbiased estimator built on it. Estimation usually takes for granted that a parameter that is hard to estimate is \emph{faint}. Its signal is weak, its Fisher information small but nonzero, and the remedy is more statistics, a richer model, or a sharper estimator. A \emph{coherent over-rotation} is a small systematic miscalibration in which the gate rotates the qubit a fixed amount too far, the same way on every run. We write it as a signed angle, so a negative value means the gate rotates that much too little. It is a leading error on today's superconducting and trapped-ion processors, a few degrees, $\eta\sim 0.01$--$0.1\,$rad~\cite{hashim2024practical}. The cheapest data a device returns is a single fixed-basis histogram of outcomes, and the textbook response is to add a parameter for the over-rotation and fit it. That enriched fit does not recover it, and more shots (repeated runs of the circuit) do not rescue it.

The first obstruction is exact and global, because a $Z$-diagonal input and a $Z$ readout are both blind to a $Z$-conjugation. Flipping the signs of the coherent angles therefore leaves the passive histogram
\emph{identically} unchanged for all angles, rates, and $n$. An over-rotation is removed by rotating back through the same angle, so a calibration that gets the sign wrong adds the error a second time. The second obstruction is local and appears at zero angle, where a coherent over-rotation leaves the fixed-basis histogram stationary to first order and its singular Cram\'er--Rao bound (CRB) is infinite. At a generic small angle and $n=3,4$, a coherent perturbation and a matching stochastic one still move the histogram the same way, so that combination stays invisible.  A direction in the kernel of the Fisher
information leaves no first-order signature in the data, however many shots are taken. This is the
classical degeneracy of singular-information
estimation~\cite{rothenberg1971identification,stoica2001singular}, whose elementary instance is a
measurement returning only $a+b$. Standard characterization tools
sidestep it by driving the device actively, through twirling, replaying, or composing gates~\cite{nielsen2021gateset,proctor2017what,huang2022foundations}. Forward models of coherent
readout predict the effect without asking whether the passive datum determines it~\cite{malik2026coherence}. We ask what a single fixed-basis histogram identifies on its own.

Because the obstruction lies in the measurement, it is removed by changing what is measured. A design of product-Pauli settings added to $Z^n$ resolves the signed angles uniformly if and only if it \emph{covers} the support. On $|\eta_R|<\pi/2$, $0\le p_R<1/2$ coverage is constructive, returning each signed angle as a ratio of two measured expectations,
$\tan\eta_S=-\langle Y_iX_{S\setminus\{i\}}\rangle/\langle Z_i\rangle$, without knowing the stochastic rates. We prove this for an explicit multi-qubit class, commuting weight-1 and weight-2 transverse
coherent over-rotations with known support on the canonical input.

This paper makes four contributions. (1)~We prove that the passive $Z$-histogram is exactly invariant under a
sign group acting on the coherent angles, of order $2^{n+1}$ for the complete family of all single- and two-qubit generators with $n\ge2$. A sign-flipped pair generates the same data, so any estimator, randomized ones included, incurs a mean-squared error of at least a quarter of their squared separation at one of them, at any sample size. At zero angle the obstruction becomes continuous, with the coherent directions in the kernel of the Fisher information, whose closed form we give
(\Cref{sec:veil}). (2)~We prove that coverage is necessary and sufficient for signed recovery, and that for the complete family the twirl-free \emph{separating code} of $\lceil\log_2(n+1)\rceil$ fixed settings is the smallest product-Pauli augmentation of $Z^n$ achieving it (\Cref{sec:design}). (3)~We show that coverage bounds the conditioning of the closed-form inverse from
below, and that this conditioning sets the Cram\'er--Rao bound on its finite-sample cost. We give the floor $c_B$ explicitly for the complete family, a fixed-$n$ guarantee since it is
exponentially small in $n$, and prove that below coverage the worst-direction cost of a locally unbiased estimator
diverges as $\Omega(\eta^{-2})$ along the full-rank approaches of \Cref{thm:threshold} (\Cref{sec:design,sec:exp}). (4)~We observe the predicted ordering in the recovered angle magnitudes on IBM Heron hardware with injected errors (\Cref{sec:exp}). The $\lceil\log_2(n+1)\rceil$ count is inherited from group testing.

%% file: sections/1_setup.tex
\section{Setup}\label{sec:setup}
A device takes a $Z$-diagonal real input and applies a noise channel that composes commuting weight-1 and weight-2
transverse over-rotations of small angle $\eta_S$ with stochastic Pauli flips of rate $p_S$. It then returns the
computational-basis ($Z$) output histogram. This single fixed-basis histogram is the passive datum a device returns by default. The observation map sends parameters $\theta=(\etav,p)$ to this histogram. The operative quantity is then the
\emph{classical Fisher information} $\FZ$ of this one fixed $Z$ measurement, which governs the achievable
precision of \emph{any} estimator built from the passive histogram. With $J_p$ the outcome-probability
Jacobian and $p_x>0$ on the support, $\FZ=J_p^\top\mathrm{diag}(1/p_x)J_p$, and its kernel is that of the
parity Jacobian $\JZ=\partial(\text{parities})/\partial\theta$: $\KZ=\ker\FZ=\ker\JZ$. A direction in $\KZ$ is a parameter combination to which the passive histogram is
first-order blind. The quantum Fisher information upper-bounds $\FZ$. Closing the gap between the two requires a measurement-design choice, taken up in \Cref{sec:design}, because no estimator built on the passive histogram can recover what the histogram does not encode.

Write $G$ for the set of generators and each generator's $X$-support as a binary incidence vector. Then
$B_\Omega$ collects the stochastic-block columns and $H$ the coherent-block columns. \Cref{app:setup} carries the full notation, including the row scalings and the Pauli-channel eigenvalues $d_a=1-2p_a$. Here $P_B$ is the Euclidean projector onto
$\mathrm{col}(B_\Omega)$. Write $\lambda_2$ for the smallest eigenvalue of the stochastic-block Gram
$J_{\mathrm{stoch}}^\top J_{\mathrm{stoch}}$, which sets the floor when readout limits the estimate and angle
damping does not. The
conditioning of a design is measured by its smallest singular value $\sigmin(\JZ)$. The rank defect
$\rH=2|G|-\mathrm{rank}(\JZ)$ is the number of directions one histogram cannot resolve.

\paragraph{Reduction to $X$-type generators.} A conjugation lemma (\Cref{app:separation}) reduces general weight-1 and
weight-2 transverse families, whether $X$-type, $Y$-type or $XY$-mixed, to the $X$-type case. Statements
below are therefore made for $X$-type generators without loss of generality. Concretely, a pairwise-commuting transverse
family of distinct weight-1 and weight-2 supports admits a local Clifford conjugation that maps each generator
to an $X$-type Pauli. This conjugation fixes every $Z$-diagonal input and the $Z^n$ measurement. Under it the data
map is unchanged, so rank, kernel dimension, singular values, and the conditioning floor $c_B$ are all preserved
(see \Cref{app:separation} for the full argument).

Two results below extend beyond small angles: the sign degeneracy of \Cref{thm:nonest}(i), which holds
at every angle and rate, and the closed-form inverse of \Cref{thm:separatrix}(ii), which holds on the box
$|\eta_R|<\pi/2$, $0\le p_R<1/2$. The remaining
claims are exact at $\eta=0$ and in the small-angle neighborhood. Throughout, generators are weight-1 and
weight-2 transverse with known support.

\paragraph{Terminology.} The \emph{null space} $\KZ$ collects the coherent--stochastic directions a single
$Z$ histogram cannot resolve to first order. A design \emph{covers} the support when some setting puts a $Y$ on the qubit
of each singleton generator, and an $X$ and a $Y$ on the two qubits of each pair generator. Coverage is a
combinatorial property of the setting list. On the canonical input at $\eta=0$ it is equivalent to full rank. At a nonzero reference a design can reach rank without covering.
Its \emph{conditioning} is measured by the smallest singular value $\sigmin$, which sets the finite-sample cost. The \emph{separating code} $\Bstar$ is the
$\lceil\log_2(n+1)\rceil$-setting product-Pauli design that restores identifiability.

The scope of every result below is summarized in \Cref{tab:scope}. Beyond this regime, two obstructions remain open: noncommuting
weight-1 and weight-2 generators, and unknown support. We return to both in \Cref{sec:concl}.

\begin{table*}[t]
\centering
\caption{Scope of each result, its regime of validity, and its status. The sign degeneracy holds at all
angles and the closed-form inverse on the stated box. The Fisher-rank statements are exact at $\eta=0$ and in
its neighborhood. The
coherent-fraction version is checked numerically at small nonzero angles, and the non-quantum and hardware
items are stated as illustrations without proof.}
\label{tab:scope}
\small
\setlength{\tabcolsep}{4pt}
\begin{tabular}{p{0.30\linewidth}p{0.46\linewidth}p{0.18\linewidth}}
\toprule
Result & Scope of validity & Status \\
\midrule
Kernel, rank formula, basis (\Cref{thm:kernel}) & $Z$-diagonal input with nonvanishing parities, full-rank $B_\Omega$, commuting weight-1 and weight-2, known support, generic small angle & exact, checked to $n\le4$ \\
Sign degeneracy, two-point risk bound (\Cref{thm:nonest}(i)) & $Z$-diagonal input, all angles and rates, any estimator, order $2^{n+1}$ for the complete family with $n\ge2$ & proved \\
$\eta_S$ non-estimable (\Cref{thm:nonest}(ii)) & each coherent angle at $\eta=0$, $0<p_R<1/2$, for finite-variance, locally unbiased estimators & proved \\
\quad coherent fraction & small nonzero angles & checked numerically \\
Coverage necessary and sufficient for uniform signed recovery, closed-form inverse (\Cref{thm:separatrix}(i),(ii)) & canonical input, $|\eta_R|<\pi/2$, $0\le p_R<1/2$, rates unknown & proved \\
Minimality of $\Bstar$ (\Cref{thm:separatrix}(iii)) & complete family, among product-Pauli augmentations of $Z^n$ & proved \\
Floor $c_B=\min\{(1-2p)^n,\sqrt{\lambda_2}\}$ (\Cref{thm:threshold}) & complete family, canonical input, uniform $p$, parity coordinates, fixed $n$, at $\eta=0$ & proved, given explicitly \\
$\Omega(\eta^{-2})$ below coverage & full-rank approach to $\eta=0$ with outcome probabilities bounded away from zero, for finite-variance, locally unbiased estimators & proved lower bound \\
Generic rank defect $\rH^\star(n,K)$ (\Cref{app:separation}) & generic input and angle, $K\le2$ & rank constant (proved), value checked for $n\le8$ \\
Identifiability under a fixed input (\Cref{thm:separation}) & canonical input, near calibration, Pauli tomography & proved \\
Non-quantum instances & classical learning & illustration \\ Hardware & $n=3$ injected errors & consistency check \\
\bottomrule
\end{tabular}
\end{table*}

%% file: sections/2_veil.tex
\section{What the Passive Histogram Cannot Identify}\label{sec:veil}

A single computational-basis ($Z$) histogram is the most economical readout of a noisy circuit: one measurement setting, and the native output of essentially every device. \Cref{thm:nonest} gives two obstructions: an exact sign symmetry that holds at every angle, and, at the calibration point where a near-ideal device sits, a Fisher information that is \emph{singular} along the subspace of the noise parameters where the coherent (unitary) over-rotations live.

The geometry is the one drawn in \Cref{fig:veil}. Fix the local parameter plane spanned by a coherent
over-rotation and the stochastic channel it mimics. For a $Z$-only histogram the level curves run
along the coherent--stochastic \emph{null direction}. Pushing population a little coherently and a
little stochastically leaves the measured outcome distribution unchanged to first order, so the histogram
cannot resolve it. One setting from the separating code $\Bstar$, introduced in
\Cref{sec:design}, resolves the valley. It lays down a second family of level curves that cross the
null direction, so the same coherent--stochastic move now changes a measured outcome.

\begin{figure*}[t]
\centering
\includegraphics[width=\textwidth]{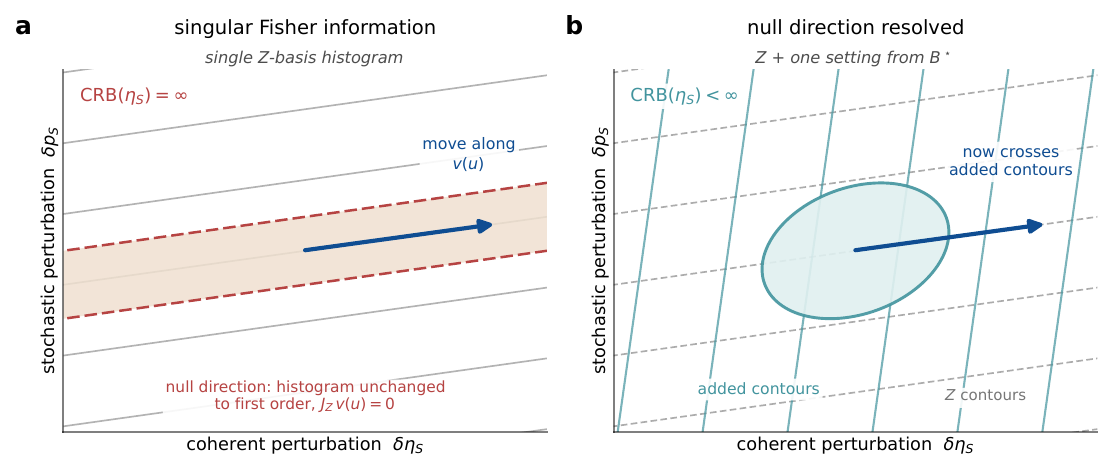}
\caption{Fisher geometry of the mechanism in the local parameter plane (schematic). \textbf{(a)} For a single
$Z$-basis histogram the level curves run along the coherent--stochastic null direction, so a move along $v(u)$ stays on a contour ($\JZ v(u)=0$), the information is singular, and $\mathrm{CRB}(\eta_S)=\infty$. \textbf{(b)} One setting from the separating code $\Bstar$ adds a second family of
contours crossing the null direction, so the same motion changes a measured outcome, the information is
full rank, and $\mathrm{CRB}(\eta_S)$ is finite. Shading marks schematic local uncertainty regions, unbounded along $v(u)$ in (a) and bounded in the plotted plane in (b). The panels draw the local, first-order face of the obstruction. The sign degeneracy of \Cref{thm:nonest}(i) is a separate, exact statement, and the added
setting drawn here resolves the sign of the generator it covers.}
\label{fig:veil}
\end{figure*}

\subsection{The kernel of the passive Jacobian}

We characterize this null space exactly. Write $\JZ$ for the Jacobian of the single-setting $Z$
histogram with respect to the noise parameters and $\KZ=\ker\JZ$ for its kernel, the locally non-identifiable directions. 

\begin{theorem}[Kernel of the passive Jacobian]\label{thm:kernel}
At a generic small-angle reference with $B_\Omega$ of full column rank and nonvanishing input parities
(both automatic for the canonical input $|0^n\rangle$), $\KZ$ admits a closed form: $\mathrm{rank}(\JZ)=\mathrm{rank}(B_\Omega)
+\mathrm{rank}((I-P_B)H)$, with an explicit coherent--stochastic null space spanned by $v(u)=(u,\tfrac12\,d\odot c(u))$ for $u\in\ker((I-P_B)H)$, where $c(u)=(B_\Omega^\top B_\Omega)^{-1}B_\Omega^\top Hu$. The rank
defect $\rH$ counts these null directions. At $n=3$ for weight-1 and weight-2, all six coherent generators lie in the null space at
$\eta=0$ ($\rH=6$), and a generic small angle lifts one combination, leaving $\rH=5$.
\end{theorem}

The proof (\Cref{app:kernel}) is constructive. Up to invertible row and column scalings, the
single-histogram Jacobian factors as $[\,H \mid B_\Omega R\,]$. The stochastic columns span
$\mathrm{col}(B_\Omega)$, so the only identifiable coherent directions are the part of $H$ outside that span. Reading off the kernel gives
the null-space basis $v(u)$ above, in which each null direction pairs a coherent move $u$ against the precise
stochastic move $\tfrac12\,d\odot c(u)$ that cancels its effect on the histogram. This is exactly the valley
of \Cref{fig:veil}(a), a null direction the histogram cannot see, and the rank defect $\rH$ is the
dimension of the null space.
At $n=3$ for weight-1 and weight-2, that span contains the entire coherent block at $\eta=0$.

\paragraph{The two-qubit case.} The smallest instructive case is $n=2$ with generators
$X_1,X_2,X_1X_2$. On $|00\rangle$ the three nontrivial $Z$-expectations are $E_1=d_1d_{12}\cos\eta_1\cos\eta_{12}$,
$E_2=d_2d_{12}\cos\eta_2\cos\eta_{12}$, and $E_{12}=d_1d_2\cos\eta_1\cos\eta_2$, with $d_a=1-2p_a$ the Pauli-channel
eigenvalues. Each coherent angle enters only through a $\cos\eta$, so at $\eta=0$ every derivative
$\partial_{\eta_a}E_T$ vanishes. The histogram then measures the three stochastic rates and nothing coherent, and
$\rH=|G|=3$. At a generic small angle the kernel mixes coherent and stochastic moves. Moving an angle $\eta_a$ by $u_a$ and
its flip rate $p_a$ by $-\tfrac12 d_a\tan(\eta_a)u_a$ leaves all three parity expectations unchanged \emph{to first order}, because each coherent angle trades against its own stochastic partner along the valley of \Cref{fig:veil}(a). Two added settings
sensitive to $Y_1$ and $Y_2$ ($\lceil\log_2 3\rceil=2$) supply the missing coherent quadratures and lift $\rH$
to $0$.

\subsection{Sign degeneracy and non-estimability}

The kernel formula is a statement about first derivatives at a reference point. The passive histogram also
carries an exact, angle-independent degeneracy, a discrete sign symmetry of the data map itself, that no amount of data touches. 

\begin{theorem}[Sign degeneracy and non-estimability]\label{thm:nonest}
Fix a $Z$-diagonal input and the $Z^n$ measurement.
\textbf{(i) Global.} For every $s\in\{\pm1\}^n$ and $g\in\{\pm1\}$ the reparametrization
\begin{equation}\label{eq:signgroup}
\eta'_S=g\Big(\textstyle\prod_{i\in S}s_i\Big)\eta_S,\qquad p'=p
\end{equation}
leaves the passive outcome distribution exactly invariant, at every angle and every $n$. The group of sign
patterns it generates has order $2^{\mathrm{rank}_{\mathbb{F}_2}M}$ with $M$ the $|G|\times(n+1)$ matrix of
rows $(1,\mathbf{1}_S)$, equal to $2^{\,n+1}$ for the complete singleton-and-pair family with $n\ge2$.
Consequently, for any estimator $\hat\etav$ of the angle vector, randomized ones included, with no
unbiasedness or regularity assumption,
\begin{equation}\label{eq:twopoint}
\max\{R(\etav),R(\etav')\}\ \ge\ \tfrac14\|\etav-\etav'\|^2
\end{equation}
at every sample size, where $R$ is the mean-squared error and $\etav'$ is any image of $\etav$ under
\eqref{eq:signgroup}.
\textbf{(ii) Local.} At $\etav=0$ the orbit collapses and \eqref{eq:twopoint} is vacuous. There the gradient
of every coherent angle lies in $\KZ$ exactly, so with every rate $0<p_R<1/2$ the singular Cram\'er--Rao bound is infinite and no
finite-variance, locally unbiased, first-order estimator recovers $\eta_S$ from the single histogram. Any
functional with a $\KZ$-component gradient is likewise non-estimable.
\end{theorem}

\noindent\emph{Proof sketch of \textup{(i)}.} Conjugating by $\prod_{i:s_i=-1}Z_i$ sends $X_S$ to
$\epsilon_SX_S$ with $\epsilon_S=\prod_{i\in S}s_i$. That flips the corresponding angles while fixing the
rates, the $Z$-diagonal input and the $Z^n$ readout. Complex conjugation in the computational basis flips
all angles at once. The two-point bound then follows because the pair induces one law
(\Cref{app:nonest}).

The separating code removes this degeneracy. Within the domain of \Cref{thm:separatrix} the inverse
\eqref{eq:inverse} returns the signed angle, so the added settings break \eqref{eq:signgroup}. 

\paragraph{The mechanism: an extended singular Cram\'er--Rao bound.}
Let $\FZ=J_p^\top\mathrm{diag}(1/p_x)J_p$ be the Fisher matrix of the single-setting histogram, whose kernel
is $\KZ$ (\Cref{lem:kerF}), and $g=\nabla\phi(\theta_0)$ the gradient of a scalar functional of the noise
parameters. The correct
singular bound is the \emph{extended} one,
\begin{equation}\label{eq:extcrb}
N\,\mathrm{Var}_{\theta_0}(\hat\phi)\ \ge\ C_F(g)=\begin{cases} g^\top\FZ^+g, & g\in\mathrm{Range}(\FZ),\\[2pt] +\infty, & P_{\ker\FZ}\,g\ne0,\end{cases}
\end{equation}
where $\FZ^+$ is the Moore--Penrose pseudoinverse. Both branches follow from one fact about the score $s=\nabla_\theta\log p_\theta$. A direction $a\in\ker\FZ$ has $\mathrm{Var}(a^\top s)=0$, so $a^\top s=0$
almost surely, and finite-variance local unbiasedness then forces the estimator gradient into
$\mathrm{Range}(\FZ)$ (\Cref{app:nonest}). A gradient with a $\KZ$ component is unattainable.

The extended bound is needed because the naive number
$g^\top\FZ^+g$ is finite for such a gradient. The pseudoinverse sets the $\KZ$ component to zero, returning the variance bound one would have if the kernel direction did not exist. Reporting it would mislead, because the finite accuracy it suggests is out of reach. The $+\infty$ branch corrects this by recording that the parameter is beyond every finite-variance, locally unbiased estimator built on this measurement when $P_{\KZ}g\ne0$.
Applied to a coherent angle $\eta_S$, \Cref{thm:kernel} gives $e_{\eta_S}\in\KZ$ for every generator $S$ at
$\eta=0$, so each coherent angle sits on the $+\infty$ branch. The normalized coherent fraction inherits the
obstruction wherever its gradient keeps a component in $\KZ$, a corollary checked numerically at the
small-angle references (see \Cref{app:nonest}).

The size of the continuous defect depends on the angle and on the number of qubits. Exactly at $\eta=0$ every coherent direction lies in the kernel ($\rH=|G|$), and a
generic small angle leaves $\rH=5$ at $n=3,4$. At generic \emph{nonzero} angles a single histogram becomes
full rank for $n\ge5$, in the range we compute, $n\le8$. Its smallest singular value falls with the angle,
so the remaining limitation there is conditioning. The sign degeneracy \eqref{eq:signgroup} is unaffected by this and
persists at every $n$ and every angle, so a full-rank passive Jacobian still leaves the signed angles
unresolved.

\subsection{The bias--variance decomposition}

\Cref{tab:biasvar} measures what the obstruction costs in finite samples. The enriched single-histogram
estimator adds the coherent parameter to the $Z$-only model but adds no new measurement. It costs an order of
magnitude or more in mean-squared error on the recovered angle magnitudes, while the separating-code design
recovers them. The composition of that error changes across the regime. At the two smallest angles it is bias-dominated, the fitted magnitudes sitting well away from the injected ones, and it
becomes variance-dominated as the angle grows.

 The separating code is the minimal
product-Pauli added-measurement design at zero angle (\Cref{thm:separatrix}). Designs that all add measurement are compared against one another by conditioning in
\Cref{sec:design} and \Cref{fig:sweep}.

\begin{table}[t]
\centering
\caption{Mean-squared error of the enriched single-histogram estimator against the separating-code
design ($n=3$, weight-1 and weight-2, every $\eta_S=\theta_{\mathrm{coh}}$ and $p_S=0.05$, $40$ seeds, $8192$ shots per design split equally across its settings, both fitted by bounded least squares from four random-sign starts at $0.1$\,rad). The enriched model does not recover the coherent error, whereas the separating-code design does. The enriched $Z$-only estimator
incurs $17$--$68\times$ the mean-squared error of the code across the regime, the $40$-seed bootstrap
intervals excluding unity at every angle (\Cref{app:threshold}). Errors are on recovered angle \emph{magnitudes}. The enriched estimator is bias-dominated at
$\theta_{\mathrm{coh}}=0.05,0.10$ and variance-dominated at $0.15,0.20$. At the tested references each coherent-angle gradient keeps a Fisher-kernel component, so the table reports one finite-sample consequence of the kernel obstruction.}
\label{tab:biasvar}
\resizebox{\columnwidth}{!}{%
\begin{tabular}{lccccc}
\toprule
& \multicolumn{2}{c}{enriched $Z$-only model} & \multicolumn{2}{c}{separating-code design} & MSE \\
\cmidrule(lr){2-3}\cmidrule(lr){4-5}
$\theta_{\mathrm{coh}}$ & bias$^2$ & variance & bias$^2$ & variance & ratio \\
\midrule
$0.05$ & $2.6\times10^{-2}$ & $8.8\times10^{-3}$ & $4.2\times10^{-5}$ & $4.7\times10^{-4}$ & $68\times$ \\
$0.10$ & $1.2\times10^{-2}$ & $1.0\times10^{-2}$ & $5.1\times10^{-5}$ & $6.4\times10^{-4}$ & $33\times$ \\
$0.15$ & $5.7\times10^{-3}$ & $9.3\times10^{-3}$ & $4.0\times10^{-5}$ & $5.6\times10^{-4}$ & $25\times$ \\
$0.20$ & $2.2\times10^{-3}$ & $8.9\times10^{-3}$ & $3.1\times10^{-5}$ & $6.0\times10^{-4}$ & $17\times$ \\
\bottomrule
\end{tabular}}
\end{table}

%% file: sections/3_design.tex
\section{Measurement Design, Conditioning, and Threshold}\label{sec:design}

A single combinatorial condition on the measurement settings resolves what \Cref{sec:veil} showed the passive histogram cannot separate. The condition is also constructive, being equivalent to the existence
of a closed-form inverse for the signed angles. A
twirl-free product-Pauli \emph{separating code} of $\lceil\log_2(n+1)\rceil$ settings meets it
(\Cref{thm:separatrix}), and the conditioning of that inverse sets the Cram\'er--Rao bound on its sampling cost (\Cref{thm:threshold}).

A design $B$ of product-Pauli settings, added to $Z^n$, \emph{covers} the support if every singleton
$\{i\}\in G$ has a setting with $b_i=Y$ and every pair $\{i,j\}\in G$ has a setting with
$\{b_i,b_j\}=\{X,Y\}$. Write $A_i=\{R\in G:i\in R\}$ for the generators incident to qubit $i$ and
$D_i=\prod_{R\in A_i}(1-2p_R)$ for the damping they impose.

\begin{theorem}[Coverage, signed recovery, and minimality]\label{thm:separatrix}
Fix the canonical input, known supports, and a channel held fixed across settings.
\textbf{(i)} If $B$ does not cover a generator $S\in G$, then at parameter points with all other coherent angles zero and the
same stochastic rates, $+\eta_S$ and $-\eta_S$ produce identical outcome distributions in every setting of
$Z^n\cup B$, so no function of the data recovers the signed angle.
\textbf{(ii)} If $B$ covers, then for $|\eta_R|<\pi/2$ and $0\le p_R<1/2$ for all $R\in G$, every signed
angle is given in closed form by expectations of $Z^n\cup B$,
\begin{equation}\label{eq:inverse}
\tan\eta_S=-\frac{\langle Y_i X_{S\setminus\{i\}}\rangle}{\langle Z_i\rangle},
\end{equation}
where $i\in S$ is the qubit carrying the $Y$ in a covering setting. The stochastic rates need not be known.
\textbf{(iii)} For the complete family with $n\ge2$, every covering $B$ has
$|B|\ge\lceil\log_2(n+1)\rceil$, attained by the binary code $\Bstar$.
\end{theorem}

\noindent\emph{Proof sketch.} The incident supports $A_i$ are linearly independent over $\mathbb{F}_2$,
so no coherent-cycle term survives and \Cref{lem:Eclosed} collapses to
\begin{equation}\label{eq:moments}
\begin{aligned}
\langle Z_i\rangle&=D_i\!\!\prod_{R\in A_i}\!\!\cos\eta_R,\\
\langle Y_iX_{S\setminus\{i\}}\rangle&=-D_i\sin\eta_S\!\!\prod_{R\in A_i\setminus\{S\}}\!\!\cos\eta_R .
\end{aligned}
\end{equation}
The damping and the spectator cosines are common to both, so their ratio is $-\tan\eta_S$ and carries no rate. On the stated box $\langle Z_i\rangle>0$, so the principal branch returns the signed angle, and
coverage is exactly the statement that each numerator is measured. Necessity in \textup{(i)} follows because
an uncovered generator contributes no term odd in $\eta_S$ to any measured row, and \textup{(iii)} is a
count of $Y$-signatures. \Cref{app:separatrix} gives both in full.\hfill$\square$

\medskip
The count $\lceil\log_2(n+1)\rceil$ is the binary separating-code size familiar from group testing and
nonrandom superimposed codes~\cite{kautz1964superimposed}. The code resolves a new object, the commuting-Pauli-channel coherent--stochastic kernel of \Cref{thm:kernel}, and comes with the conditioning floor below. The inverse \eqref{eq:inverse}
establishes \emph{coherent} identifiability throughout the box. Joint identifiability of the rates follows
near calibration from \Cref{prop:jointrank}, and away from it from the stochastic-block rank at the
operating point. A logarithmic set of fixed settings thus recovers, in closed form, the quantity that active multi-circuit protocols are built to recover. Two non-Pauli product settings, measuring $(X\pm Y)/\sqrt2$ on every qubit and added to $Z^n$, determine every signed angle of the complete family at any $n$.
Their difference measures $\langle X_iY_j+Y_iX_j\rangle$, which \eqref{eq:moments} gives as
$-(\langle Z_i\rangle+\langle Z_j\rangle)\tan\eta_{ij}$, and the singletons follow the same way. The logarithmic count is thus the floor within the Pauli alphabet.

\paragraph{The price of the inverse.} Coverage is a yes-or-no property of the setting list, and
\eqref{eq:inverse} makes it an exact identity on population expectations. Finite data enter the ratio only
through empirical means, and there the damping that cancels between numerator and denominator does not
cancel from the fluctuations. A delta-method expansion at calibration gives
\begin{equation}\label{eq:deltavar}
\mathrm{Var}(\hat\eta_S)=\frac{1}{N_Q D_i^2}+o(N_Q^{-1})
\end{equation}
for $N_Q$ shots on the covering setting (\Cref{app:threshold}). The angular sensitivity of the inverse is
therefore the damping product $D_i$, which is $(1-2p)^n$ for the complete family at uniform rate. This is the first
branch of the floor below. Estimating the rates jointly adds the stochastic block, and the governing number becomes the design's smallest singular value $\sigmin$, which brings in the second branch $\sqrt{\lambda_2}$.

Conditioning separates designs that all pass an identifiability test
(\Cref{fig:cond}a). In matched-budget simulations, at equal cardinality and total shots, full
rank at the tested reference is easy to come by. Equal-size random product-Pauli designs reach $\rH=0$ in
$54\%$ of trials at $n=3$ and $\approx100\%$ at $n=4$--$6$. Rank at a nonzero reference is a weaker condition than coverage, and the difference shows in the conditioning. A stratified design that leaves the singleton
$\{1\}$ uncovered still reaches $\rH=0$ at the tested reference, but only through higher-order mixing. Its
smallest singular value is $0$ at $\eta=0$ and $\sigmin\approx7\times10^{-4}$ at the reference, and it incurs
up to $6\times$ the finite-sample cost (root-mean-square error, RMSE, $0.140$ versus $0.023$ at $n=3$). The separating code covers,
so it attains the floor by construction, with $\sigmin=c_B=0.73$ at $\eta=0$ (here $n=3$,
$p=0.05$, so $0.9^3\approx0.73$), and $\sigmin\approx0.72$ at the generic reference
(\Cref{tab:coords}). It matches a greedy log-det $D$-optimal search at $n=4$ and exceeds it at $n=3$, where
the greedy design reaches only $\sigmin=0.097$.

The association extends across a sampled population of designs passing the same reference-rank filter,
where $\sigmin$ tracks finite-sample RMSE and the worst-conditioned designs cost several times the RMSE of the best at equal shots (\Cref{sec:exp}). The $\sigmin$ values quoted here are closed-form parity-coordinate singular values. At the peaked
$|0^n\rangle$ reference they track the exact multinomial Fisher conditioning, in that the parity $\sigmin$ and the Fisher $\sigmin$ rank this population almost identically and the Fisher $\sigmin$ predicts
RMSE comparably well (\Cref{app:threshold}).

\begin{figure*}[t]
\centering
\includegraphics[width=\textwidth]{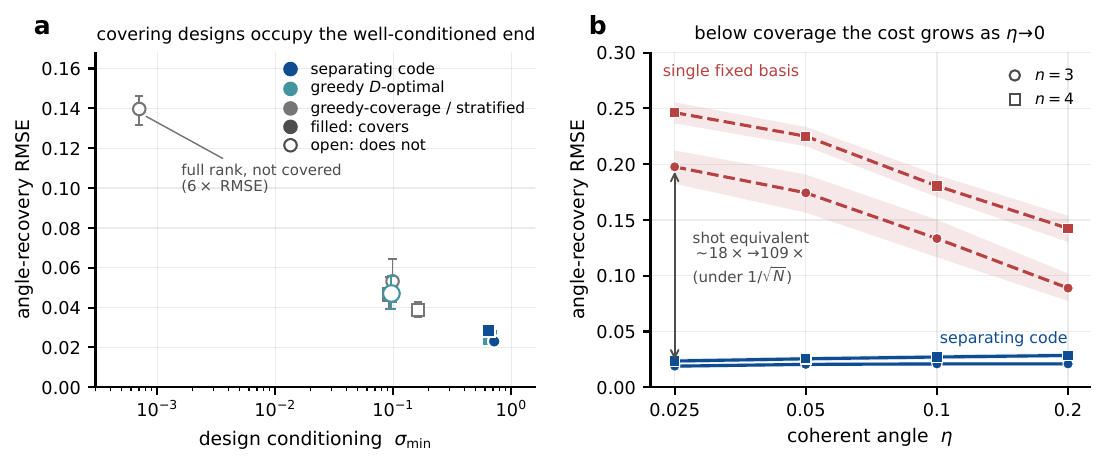}
\caption{Coverage versus conditioning at matched budget. Error bars are $95\%$ bootstrap intervals over
$50$--$60$ seeds. \textbf{(a)} Named designs at $n=3$ (circles) and $n=4$ (squares). Marker fill is the structural
coverage condition, a property of the setting list evaluated at $\eta=0$. The horizontal position is
$\sigmin$ in parity coordinates at the nonzero reference used throughout ($\eta=0.1$, $p=0.05$), which is
why designs that fail to cover still appear at a positive $\sigmin$. Every plotted design reaches full rank
there, yet finite-sample angle-recovery RMSE tracks $\sigmin$. The separating code sits at the well-conditioned end, matching a greedy $D$-optimal search at $n=4$ and exceeding it at $n=3$,
while a design that reaches rank without covering the support costs $\sim6\times$ the RMSE. Parity-coordinate
$\sigmin$ and Fisher conditioning rank this population almost identically (\Cref{app:threshold}).
\textbf{(b)} Below coverage, a single fixed basis incurs a finite-sample cost that grows as the coherent angle
$\eta$ shrinks, while the separating code stays $\eta$-flat. The shot equivalent
$(\mathrm{RMSE}_{\text{single}}/\mathrm{RMSE}_{\text{code}})^2$ under $1/\sqrt N$ scaling grows
$\sim18\times\!\to\!109\times$ ($n=3$) as $\eta$ falls $0.2\!\to\!0.025$, the direction
\Cref{thm:threshold} predicts below coverage. Its $\Omega(\eta^{-2})$ rate is asymptotic, and these fixed-shot fits do not measure it.}
\label{fig:cond}
\end{figure*}

Calibration drives the coherent angles toward zero, where the two regimes separate most sharply. A design below coverage loses sensitivity as $\eta\to0$, while a covering design keeps a conditioning
floor that depends only on $n$ and the admissible rates.

\begin{theorem}[Conditioning floor and threshold]\label{thm:threshold}
For the complete weight-1 and weight-2 family on the canonical input at uniform rate $p$, the separating
code's conditioning floor
in parity coordinates is $c_B=\min\{(1-2p)^n,\sqrt{\lambda_2}\}$. Below coverage, along an approach $\etav\to0$ on which the design stays full rank and the outcome
probabilities remain bounded away from zero, the worst-direction sample complexity diverges as
$\Omega(\eta^{-2})$ for finite-variance, locally unbiased estimators. Above coverage, $\sigmin$ is bounded away from zero uniformly in small $\|\etav\|$.
\end{theorem}

\noindent\emph{Proof sketch.} At $\eta=0$ the Jacobian is block-diagonal across Pauli-row subspaces, because the stochastic derivatives occupy the $Z$-rows and each covered generator is read on rows carrying no other generator. Hence $\sigmin=\min\{\sigma_{\mathrm{coh}},\sigma_{\mathrm{stoch}}\}$. At a uniform rate a generator read on $m_S$ rows gives a
coherent column of norm $\sqrt{m_S}(1-2p)^n$, and the binary code reads some generator once, so
$\sigma_{\mathrm{coh}}=(1-2p)^n$, and the stochastic block contributes $\sqrt{\lambda_2}$, giving $c_B$.
Coverage supplies such a column for every generator, and any $Z$-type rows a design adds only enlarge the
stochastic Gram, so at $\eta=0$ every covering design has $\sigmin\ge c_B$. Analyticity of $J(\etav,p)$ keeps at least half the floor on an explicit neighborhood. Below coverage an uncovered column expands as $O(\|\etav\|)$.
\Cref{app:threshold} supplies the constants, the heterogeneous-rate column norms, the explicit radius, and
the hypothesis the divergence needs.\hfill$\square$

\medskip
Which branch attains the minimum depends on $n$. At $p=0.05$ the angular branch does, giving $c_B=0.73$ at $n=3$ and
$0.21$ at $n=15$, verified against the exact stochastic block through $n=18$. The branches cross near
$n=20$. Beyond it $\sqrt{\lambda_2}$ is the smaller branch and falls faster, to $\le9.42\times10^{-6}$ at $n=63$ against
$(1-2p)^{63}=1.3\times10^{-3}$. Both are exponentially small in $n$, so $c_B$ is a fixed-$n$, small-$p$
guarantee (\Cref{app:threshold}). At $\eta=0$ the floor holds for every covering design. Each such design has the same $Z$ stochastic block and coherent columns of norm at least $D_S$, so $\sigmin\ge c_B$, with equality for minimum-cardinality $XY$ codes
such as $\Bstar$. In the design population of \Cref{fig:cond}, the designs far below $c_B$ fail to cover and reach rank through higher-order mixing.

The two ``below'' and ``above'' regimes refer to the hard $\eta=0$ rank boundary. The
$\Omega(\eta^{-2})$ statement is the separate small-angle \emph{rate} at which a sub-coverage design's cost
diverges (\Cref{fig:cond}b). The link between conditioning and error is a non-asymptotic inequality, established
below. For any locally unbiased estimator built from $N$ samples of each measured setting, the worst
estimable direction obeys $N\cdot\mathrm{MSE}\ge2^np_{\min}\,\sigmin^{-2}$ in parity coordinates, with
$p_{\min}$ the smallest outcome probability. The constant $2^np_{\min}$ equals $1$ at a uniform reference (\Cref{app:threshold}).
Consequently, along an approach $\etav\to0$ on which a sub-coverage
design stays full rank, $\sigmin(J)=O(\|\etav\|)$ forces $N\cdot\mathrm{MSE}=\Omega(\eta^{-2})$ (for
finite-variance, locally unbiased estimators), so its worst-direction cost diverges as the angle shrinks. The separating code sits at $\sigmin=c_B$ at $\eta=0$, and within a factor
$2$ of it on the small-angle radius. Its worst-direction parity-coordinate bound is therefore capped at
$4c_B^{-2}$, finite at fixed $n$.

Because the code is a closed-form formula, designing measurements for this family reduces to covering the noise with $\lceil\log_2(n+1)\rceil$ settings, which puts the conditioning at or above $c_B$ at $\eta=0$. 

\begin{figure*}[t]
\centering
\includegraphics[width=\textwidth]{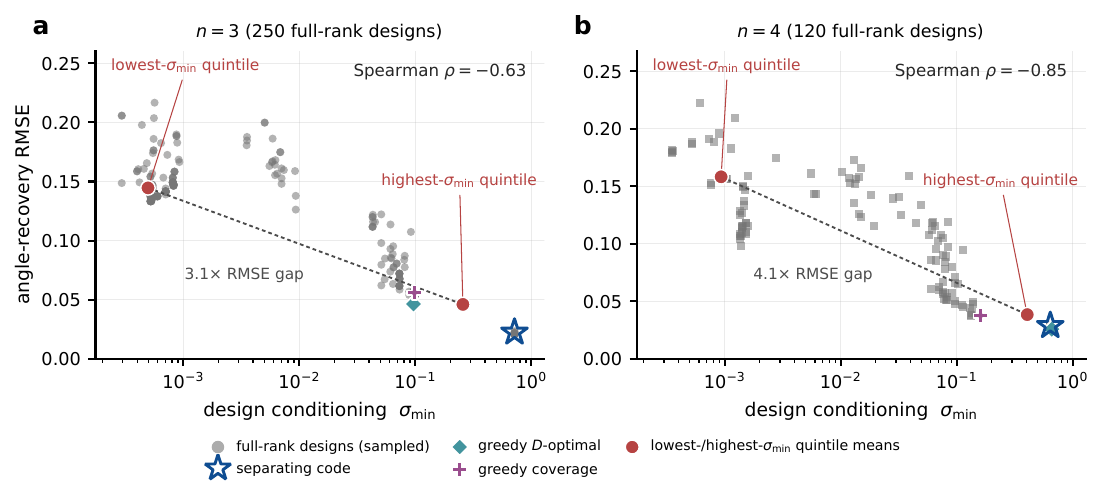}
\caption{Finite-sample recovery cost versus conditioning across the sampled design population. For $n=3$
($250$ draws over $94$ distinct designs) and $n=4$ ($120$ over $115$) we plot every sampled random product-Pauli design that reaches
full rank ($\rH=0$) at the tested reference and the matched budget $L=\lceil\log_2(n+1)\rceil$, as angle-recovery RMSE at matched shots against the design's smallest singular value $\sigmin$ in parity coordinates. The dashed segment joins the lowest- and highest-$\sigmin$ quintile means and is not a fit. The
separating code (blue) sits at the well-conditioned end by construction.}
\label{fig:sweep}
\end{figure*}

%% file: sections/4_input.tex
\section{An Input-Side Barrier}\label{sec:input}

The separating code of \Cref{sec:design} removes the null space of \Cref{sec:veil} by changing the measurement. The probe state is also a design variable, and for some noise families it is the only one that helps. Locating the obstruction on the measurement or the input side of the data map tells the experimenter whether a richer measurement, a different input, or no passive intervention at all can restore identifiability.

The transverse families of \Cref{sec:veil} are an instance where the measurement is the
lever. The singularity is a property of the histogram, and the separating code lifts it.
The contrasting instance is a pure longitudinal $Z$-channel on the canonical input. Here
no measurement helps, because the channel fixes the input state and returns the input histogram unchanged. The missing column originates in the chosen probe. Changing the input from $|0^n\rangle$ to $|{+}^n\rangle$, with the readout rotated accordingly, restores identifiability. The two cases together show that the Fisher matrix $F_Z$ inherits its rank defect from either factor.

\begin{theorem}[Identifiability under a fixed input]\label{thm:separation}
On the canonical input and near calibration, Pauli tomography (readout in every single-qubit product-Pauli setting)
identifies a weight-1 and weight-2 transverse class iff its $X$-supports are nonzero and pairwise distinct.
The canonical input matters, because the channel is unital and the maximally mixed state therefore carries no information at
all (\Cref{app:separation}). Longitudinal $Z$-noise, by contrast, is unidentifiable from $|0^n\rangle$ by \emph{any}
measurement, because the channel fixes $|0^n\rangle$. A different input ($|{+}^n\rangle$) recovers it. The
obstruction can therefore lie on the input side.
\end{theorem}

\paragraph{Proof sketch.}
The two halves of the statement isolate the two factors of the data map. For a
weight-1 and weight-2 transverse class read out by Pauli tomography, identifiability
reduces to a combinatorial condition on the $X$-supports. Equal or empty supports
collide in the incidence matrix $B$ and leave a null direction in the stochastic block.
Distinct nonzero supports instead make $B$ full column rank, and the coverage
construction of \Cref{sec:design} supplies the complementary coherent block. So in the
transverse case the obstruction is read out by the measurement, and a richer measurement
removes it. For the longitudinal case the obstruction is structural in the probe. A pure
$Z$-channel fixes every $Z$-diagonal state for all $(\etav,p)$, so every column of the
$|0^n\rangle$ Jacobian vanishes and no measurement (product, entangled, or adaptive)
can identify it. Conjugating the channel by $H^{\otimes n}$ is equivalent to preparing the input $|{+}^n\rangle$ and rotating the readout by $H^{\otimes n}$, and this map carries the longitudinal theory to the
transverse theory, which is identifiable.  The full argument, including the conjugation reduction, is in
\Cref{app:separation}.

%% file: sections/5_experiments.tex
\section{Verification and Hardware}\label{sec:exp}

\paragraph{Comparison with exact computation.}
The closed-form statements agree with exact density-matrix computation at $n\le4$ (kernel and null-space
basis) and through $n\le8$ (rank and conditioning summaries), at the differentiation step sizes and rank
thresholds of \Cref{app:setup}. The rank counts are unchanged across the tested real, complex, pure, and mixed inputs,
where the generic rank defect is $\rH^\star(3,2)=5$. At $K=1$ the canonical input is the exception, with $\rH=|G|$ (\Cref{lem:n3}). The logarithmic
measurement budget of \Cref{thm:separatrix} is confirmed by exact computation through
$n=8$, where the minimal separating set has size $|\Bstar|=2,3,3,3,3,4$ for
$n=3,\dots,8$ (the $\lceil\log_2(n+1)\rceil$ law). The angular branch of the conditioning
floor $c_B=\min\{(1-2p)^n,\sqrt{\lambda_2}\}$ is evaluated analytically to $n=63$, and the stochastic branch
$\sqrt{\lambda_2}$ is computed exactly through $n=18$, where the two have not yet crossed. These checks confirm the constants of the theory, and the remaining experiments probe what they
imply for estimation at finite shots.

\paragraph{Signed recovery.} \Cref{thm:nonest}(i) puts the sign beyond the passive histogram at every
budget, and \Cref{thm:separatrix}(ii) returns it in closed form once the design covers. \Cref{fig:sign}
separates the two at $n=3$, with $|\eta_S|$ uniform on $[0.08,0.20]$, independent signs, and heterogeneous
rates $p_S$ uniform on $[0.02,0.08]$. The passive estimator receives the whole shot budget on $Z^3$; the
code splits the same total across its three settings. Each of the $16$ sign-group elements leaves the $Z^3$ outcome distribution unchanged, exactly by \Cref{thm:nonest}(i) and to $0$ in double precision here. At finite shots we fit the single histogram by bounded least squares,
without a sign preference and with each $p_S$ a free nuisance parameter. It sits at chance over a
$64\times$ range of budgets, its sign error rate averaging $51\%$ (\Cref{fig:sign}a). Its \emph{magnitude} RMSE falls by a factor $1.8$ over that
range, $0.181$ to $0.099$\,rad, while its signed-angle RMSE has a log--log slope of only $-0.08$. Neither
quantity is identifiable from the passive histogram alone. The magnitudes are
confounded with the unknown rates along the family generated by the kernel of \Cref{thm:kernel}. Knowing the rates removes this confound. A generic nonzero angle removes it only once the parameter count permits, at $n\ge5$ in the
computed range, while a five-dimensional defect survives at $n=3$ and $n=4$ (\Cref{app:kernel} gives two
parameter points with different magnitudes and identical passive data at $n=3$). The signs are protected by the exact symmetry of
\Cref{thm:nonest}(i), which nothing about the rates lifts. The closed-form inverse
\eqref{eq:inverse}, which uses no knowledge of the rates, made no sign error in any of the $60$ trials at
each budget from $4096$ shots on, and its signed-angle RMSE falls with slope $-0.496$, the
$N_{\mathrm{tot}}^{-1/2}$ of a statistically limited estimator.

\begin{figure}[!htb]
\centering
\includegraphics[width=\columnwidth]{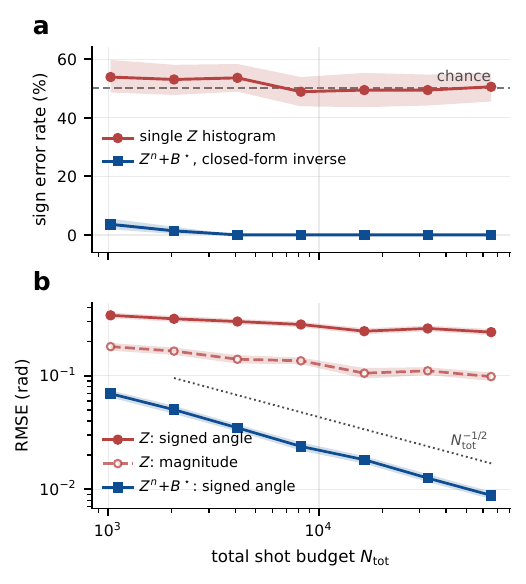}
\caption{The passive fit's magnitude error falls with the shot budget while its sign error does not. Here $n=3$, complete
weight-1 and weight-2 family, $|\eta_S|$ uniform on $[0.08,0.20]$ with independent signs, rates $p_S$
uniform on $[0.02,0.08]$, $60$ trials per budget. Bands are $95\%$ bootstrap intervals over trials. \textbf{(a)} Sign error rate against
total shot budget. The passive fit stays at chance at every budget, as \Cref{thm:nonest}(i) implies under this
sign-symmetric ensemble. The closed-form inverse \eqref{eq:inverse} makes no sign error beyond the two smallest budgets in
this regime. \textbf{(b)} The passive signed-angle RMSE is nearly flat while its \emph{magnitude} RMSE
falls, and the inverse's signed error follows $N_{\mathrm{tot}}^{-1/2}$.}
\label{fig:sign}
\end{figure}

\paragraph{Matched-budget comparison.}
We match cardinality and give each design a total of $8192$ shots, split equally across its settings and
averaged over $50$ sampling seeds. The separating code then reaches $\rH=0$ with a well-conditioned parity
Jacobian, $\sigmin\approx0.5$--$0.7$.
Equal-size random product-Pauli designs reach full rank in $54$--$100\%$ of
trials, but at a far worse conditioning, with median $\sigmin\approx5\times10^{-4}$ at
$n=3$. A greedy $D$-optimal search matches the code's conditioning at $n=4$, but at $n=3$ its log-determinant objective lands it $7\times$ below the code on $\sigmin$ ($0.10$ vs $0.72$, \Cref{fig:cond}a). An exhaustive E-optimal search over all product-Pauli augmentations targets $\sigmin$ directly. It puts
the fixed code within $0.4\%$ of the optimum at $n=3$ ($0.717$ against $0.720$), and $22\%$ below it at
$n=4$ ($0.642$ against $0.823$, attained by $\{XYYZ,YYZX,ZYXY\}$). The code is therefore not the
$\sigmin$-maximizer. It is nonetheless beaten by $1$ of the $206$ full-rank designs at $n=3$ and by $805$ of
the $85\,304$ at $n=4$, under $1\%$ in both cases, and it comes from a closed-form construction. Under the code, the
finite-sample magnitude RMSE is $\eta$-flat ($\approx0.02$--$0.03$) as the injected
angle scale $\eta$ shrinks, while a sub-coverage single setting's RMSE grows as $\eta$ falls. The same gap can be read as a shot equivalent under $1/\sqrt N$ scaling. On that reading the budget a
single setting would need to match the code grows from $\approx18\times$ to $\approx109\times$ at $n=3$, and
from $\approx25\times$ to $\approx110\times$ at $n=4$, as $\eta$ falls $0.2\!\to\!0.025$
(\Cref{fig:cond}b). The conversion is heuristic and, once the
single setting is bias-limited, optimistic, so the evidence is the direct RMSE comparison. The growth is the direction \Cref{thm:threshold} predicts below coverage. Its $\Omega(\eta^{-2})$ rate is
asymptotic and concerns finite-variance, locally unbiased estimators. These are bounded least-squares fits
at fixed shots on a rank-deficient single setting, outside the theorem's full-rank hypothesis, so the curves illustrate the trend.

The same ordering appears across a sampled design population. \Cref{fig:sweep} plots every random
product-Pauli design that reaches full rank at the tested reference and the matched budget
$L=\lceil\log_2(n+1)\rceil$. Designs are drawn with replacement from the product-Pauli and $XY$ pools, giving $250$ draws over $94$ distinct designs at $n=3$ and $120$ over $115$ at $n=4$, so the intervals below treat each draw as the unit. RMSE tracks $\sigmin$ across that
population (Spearman $\rho=-0.63$ at $n=3$, $-0.85$ at $n=4$; \Cref{app:threshold}). The
lowest-$\sigmin$ quintile incurs $3.1$--$4.1\times$ the RMSE of the highest-$\sigmin$ quintile at equal budget
and shots. The separating code sits at
the well-conditioned end, matching the greedy $D$-optimal search at $n=4$. Among designs that pass the same rank filter, the cost tracks conditioning.

\paragraph{Hardware consistency check.}
We tested the conditioning ordering on hardware, running two IBM Heron processors with
injected known errors at $n=3$. Before the estimator comparison, chains were screened for recovery
stability on a separate eight-circuit probe, on which the separating-code fit had to recover the injected angles with
intercept within $\pm0.02$\,rad and slope in $[0.85,1.15]$. Two further backends were excluded for fit
drift before this comparison, and their results are summarized in \Cref{app:hardware}. The screen is evaluated on
the code's own recovery, so the retained chains are ones on which the code is stable. The same ordering
holds on the excluded, drifting chains, which the screen did not retain, where the ratio stayed above one
throughout (raw ratios $1.8$--$7.9$). We report
five replicates on the retained backends (Heron A, $3$; Heron B, $2$). The metric is the median over the six
coherent generators of the magnitude error $\lvert|\hat\eta_S|-|\eta^{\mathrm{inj}}_S|\rvert$, so the run tests the conditioning ordering and leaves signed recovery to the simulation of \Cref{fig:sign}. Across the replicates the single fixed-basis ($Z$)
least-squares fit carries $3$--$5\times$ the magnitude error of the separating-code fit, a ratio of $2.97$--$5.05$ with median $3.96$ (\Cref{fig:hw}). The code recovers the injected angles with near-unit slope
($0.88$--$1.05$) and magnitude error $\le0.019$\,rad. That comparison gives every circuit the same shots, so
the code receives three times the total. Subsampling it to a matched total leaves a $2.8$--$4.0\times$ gap (\Cref{app:hardware}). At $\eta=0$, \Cref{thm:threshold} makes a single fixed basis far more weakly conditioned than the code, and the gap is also present at the injected angles. There the single basis remains rank-deficient, with $\rH=5$ by exact computation and smallest
\emph{nonzero} singular value only $1.3\times10^{-4}$ to $8.1\times10^{-3}$ across the grid. The code is
full rank ($\rH=0$) and well-conditioned ($\sigmin\approx0.72$). These are computed design properties, and the measured ratios follow the same \emph{ordering} on every replicate, though the size of the gap varies with backend and replicate.

\begin{remark}[The singular-Fisher diagnostic beyond quantum noise]\label{rem:nonquantum}
The singular-Fisher diagnostic at the core of these results extends beyond coherent quantum noise. A target whose gradient has a component in the kernel of the averaged Fisher information is beyond every finite-variance, locally unbiased estimator built on that data. A measurement design that brings the gradient into the range of the Fisher information restores the target, and conditioning then sets the finite-sample cost. The diagnostic is the classical singular-Fisher
non-identifiability of experiment design~\cite{rothenberg1971identification,
stoica2001singular}, read in terms of coverage and conditioning. Two
non-quantum instances in \Cref{app:ml}, a subspace-restricted active-learning problem and an
overparameterized random-feature probe (\Cref{fig:c7}), reproduce the three signatures.
\end{remark}

%% file: sections/6_related.tex
\section{Related Work}\label{sec:related}

Prior quantum-characterization methods obtain identifiability through active access: designed gate sequences,
twirls, random circuits, or adaptive probes. 

\paragraph{Characterization by active intervention.}
Gate-set tomography reconstructs gates from structured sequences~\cite{nielsen2021gateset}. Randomized benchmarking extracts an average error rate from random Clifford
sequences of growing depth~\cite{proctor2017what}. Pauli-channel learning achieves favorable sample
complexity through twirled or Clifford-access protocols~\cite{huang2022foundations,chen2022quantum}. At
scale, sparse Pauli--Lindblad learning~\cite{vandenberg2023probabilistic}, efficient Pauli-channel learning~\cite{harper2020efficient}, and cycle benchmarking~\cite{erhard2019characterizing} estimate the stochastic
Pauli-rate structure on hardware. Twirling a gate into a Pauli channel removes the coherent-angle vector we recover (\Cref{app:twirl}). Protocols built to recover coherent angles directly,
robust phase estimation~\cite{kimmel2015robust}, long-sequence gate-set tomography~\cite{blumekohout2017demonstration}, and Hamiltonian learning~\cite{wiebe2014hamiltonian}, reach the same
per-generator information through active, variable-depth experiments. 

\paragraph{The same target, compared.} A resource comparison is meaningful only against protocols that reach
the coherent part of the noise. Reference~\cite{crupi2025efficient} recovers
low-degree channel coefficients, coherent terms included, from \emph{random} single-qubit preparations and
\emph{random} Pauli-basis measurements with classical inversion, at logarithmic sample complexity.
Reference~\cite{kaufmann2025coherent} estimates the coherent part of a gate layer with an estimate robust to the Pauli
component, from an active gate-layer experiment. Reference~\cite{chen2026gateset} targets a different object, the
self-consistently learnable part of gate-set \emph{Pauli} noise, from Clifford access and gate repetition,
and is included here because its unlearnable remainder is the closest published analogue of our sign group.
 We fix the access to the canonical input $|0^n\rangle$, one fixed readout, and a list of static settings, which excludes randomization, twirling, gate repetition, and depth sequences. In that model we prove that coverage is \emph{necessary} and sufficient for signed recovery, and invert in closed form without knowing the rates.

The number of added settings grows logarithmically with $n$, while the shot cost is set by $c_B$, which is exponentially small in $n$. The worst-direction parity bound therefore grows at least as $(1-2p)^{-2n}$, and past
the crossing near $n=20$ at $p=0.05$ faster still. In return, the fixed-setting design works in a strictly more restricted access model, where the minimum number of settings can be proved and the inverse written down. Which protocol is cheaper at a common target accuracy remains open.

\paragraph{Sloppiness, unitarity, and forward models.}
In metrology the continuous half of the degeneracy appears as \emph{sloppiness} of the quantum Fisher
matrix, liftable by scrambling the probe or tailoring the encoding~\cite{frigerio2024overcoming,mihailescu2025metrological}. Measurement choice is known to govern how
efficiently noise parameters are recovered~\cite{tsang2023quantum}. We analyze instead the classical Fisher information of a fixed passive
measurement, the operative quantity before any optimization. We give its kernel in closed form where prior
diagnoses were numerical, and restore the lost directions with a minimal discrete product-Pauli design. The closest existing quantifier of noise coherence is unitarity benchmarking~\cite{wallman2015estimating}, which reports a single scalar magnitude from an active randomized protocol, whereas we recover the signed per-generator angle vector, which no magnitude-valued quantity determines. Forward models
of coherent readout~\cite{malik2026coherence} predict what coherent noise does to the data, and we answer the inverse question of whether the angle can be recovered from it. On the estimator side, the estimate of~\cite{kaufmann2025coherent} has the same insensitivity to the stochastic rates as \eqref{eq:inverse}.

\paragraph{Classical identifiability and experiment design.}
Our setting combines identifiability under a singular Fisher information~\cite{rothenberg1971identification,stoica2001singular,liyeh2011singular} with measurement design as optimal experimental design~\cite{chaloner1995bayesian}. A finite-sample spectral threshold for passive
identifiability over general models~\cite{huang2025spectral} is the closest abstract statement. $E$- and $A$-optimality already treat conditioning, and the closed-form floor $c_B$ for the commuting-Pauli coherent--stochastic kernel is new. On the quantum-learning
side the same active stance dominates: low-rank matrix sensing from a designed ensemble~\cite{lang2025unified}, adaptive Bayesian experiment selection~\cite{granade2012robust}, and
compressed-sensing Lindbladian tomography~\cite{dobrynin2025compressed}. Our separating code inherits its
covering count from nonrandom superimposed codes and group testing~\cite{kautz1964superimposed}, and it must also control a Fisher conditioning that separation alone leaves free.

\paragraph{Explicit distinctions.}
Gate-set tomography carries a gauge freedom that a convention or a reference frame removes~\cite{merkel2013selfconsistent,nielsen2021gateset}. Here preparation and measurement are fixed and trusted,
and the sign-related channels are physically distinct. Within the domain of \Cref{thm:separatrix} the code
separates them, and they call for opposite corrective rotations about the lab's calibrated axes. Pauli-noise learning has unlearnable degrees of freedom of its own~\cite{chen2022quantum,chen2026gateset}. They are combinations of stochastic rates, fixed by the cycle
structure of a gate set probed by gate repetition, and they act as a gauge that leaves error-mitigated
predictions unchanged. Ours acts on the coherent angles of one fixed-basis histogram, it
changes the corrective rotation a calibration would apply, and a covering design removes it.
Efficient Pauli-rate estimation targets the twirled model~\cite{flammia2020efficient}, from which twirling has already removed the coherent angles we recover. Classical shadows gain breadth by randomizing~\cite{huang2020predicting}, using random product measurements to estimate many state properties from one dataset.
The small fixed-setting random designs in our population study (\Cref{fig:sweep}) are of that family, and
they already span orders of magnitude in conditioning, whereas the $\lceil\log_2(n+1)\rceil$-setting code has a proved conditioning floor. Unbiasedness and effective-dimension
criteria for quantum estimation~\cite{kwon2025criteria,kwonlie2026effective} bound what estimators can
achieve. Our result names a coherent parameter that no finite-variance, locally unbiased estimator recovers
at $\eta=0$, and a sign that no estimator resolves uniformly.

The single-qubit coherent--stochastic degeneracy is folklore. We contribute its multi-qubit correlated
structure, the exact sign symmetry of the passive data map, the closed-form kernel and inverse, and
the minimal design with its conditioning floor, placing the result in the quantum-learning identifiability
and sample-complexity line~\cite{du2024linearproperties}.

%% file: sections/7_conclusion.tex
\section{Conclusion}\label{sec:concl}

We asked what a single passive histogram identifies about coherent noise, a property of the observation map
that is settled before any estimator is chosen. For commuting weight-1 and weight-2 transverse errors the answer differs for the magnitudes and the signs. The magnitudes are confounded with the
unknown stochastic rates along a continuous family, which a generic nonzero angle lifts only where the
parameter count permits, at $n\ge5$ in the range we compute. No individual sign is determined at all, because the sign flips form an exact symmetry group of the passive data, of order $2^{n+1}$
for the complete family with $n\ge2$, and every sign pattern in one orbit of that group gives the same data, at every angle and every rate. At the calibration point the obstruction sharpens into a singular Fisher information
with an infinite Cram\'er--Rao bound.

In simulation the passive fit is at chance on the signs at every budget we tested, while its magnitude
error falls with shots. Once a design covers, the signs come back. On $|\eta_R|<\pi/2$, $0\le p_R<1/2$ each one follows from a
ratio of two measured expectations, and $\lceil\log_2(n+1)\rceil$ added settings suffice. From there the Cram\'er--Rao cost is governed by conditioning, whose floor we give explicitly. Below coverage the worst-direction cost of a locally unbiased estimator diverges as $\Omega(\eta^{-2})$ along the full-rank approaches of \Cref{thm:threshold}. Two designs that both
pass a rank test at a nonzero reference can differ severalfold at matched budget, and on hardware the recovered angle magnitudes follow the conditioning ordering.

\paragraph{Scope.} Both branches of the floor are exponentially small in $n$, so it is a fixed-$n$,
small-$p$ guarantee. The Fisher-rank statements are local, and at generic angles the continuous degeneracy lifts, leaving for the complete family a parameter-counting deficit that matches $\max(0,2|G|-(2^n-1))$ in every case we compute ($n\le8$) and vanishes for $5\le n\le8$, while the sign degeneracy persists. Noncommuting generators and unknown support both need new theory. Unknown support needs discovery as well as design (\Cref{app:misspec} tests single-generator support errors).

\paragraph{Outlook.} The diagnostic behind these results, whether the target's gradient lies in the range of the observation map's Fisher information, applies beyond coherent quantum noise (\Cref{rem:nonquantum}). On a processor it takes a concrete form, because the device exposes a schedule of Pauli-basis settings, the primitive that Pauli grouping and basis-change compilation already build on, and a control stack can therefore check coverage on its own schedule before it fits anything. We expect the check to be useful wherever a characterization routine adds model parameters faster than measurement settings.

%% file: sections/statements.tex
\section*{Acknowledgment}
This research used resources of the Oak Ridge Leadership Computing Facility, which is a DOE Office of Science User Facility supported under Contract DE-AC05-00OR22725. This research was supported by the UTC Applied AI \& Industry Partnerships Initiative, funded through the Center of Excellence in Applied Computational Science and Engineering (CEACSE). This study was also funded by the College of Engineering and Computer Science of UTC. The funders played no role in study design, data collection, analysis and interpretation of data, or the writing of this manuscript.

%% file: appendix.tex
% ======================================================================
%  Appendix: proofs, hardware protocol, non-quantum instances
% ======================================================================
\section{Notation and standing conventions}\label{app:setup}

We collect the conventions used throughout the proofs: the coherent--stochastic noise model, the $Z$-histogram
data map, and the Pauli--Liouville facts that drive every rank computation.

\paragraph{Objects.} On $n$ qubits, for $S\subseteq[n]$ let $X_S=\bigotimes_{i\in S}X_i$, and similarly
$Z_T$. The \emph{generator set} $G$ is a set of distinct nonempty $X$-type masks of weight $1\le|S|\le K$. The parameter is $\theta=(\etav,p)\in\mathbb{R}^{2|G|}$, a coherent angle $\eta_S$ and a stochastic rate
$p_S$ per $S\in G$. The noise channel, in the fixed composition order of $G$, is
\begin{equation}
N_{\etav,p}=\prod_{S\in G}\mathrm{Flip}_{S,p_S}\circ\mathrm{Ad}_{R_S},\qquad
R_S=e^{-i\eta_S X_S/2},
\end{equation}
with the single-generator flip $\mathrm{Flip}_{S,p}(\rho)=(1-p)\rho+pX_S\rho X_S$.
The pre-noise state $\rho_0$ is real ($\rho_0=\rho_0^\ast$, so every Pauli expectation with an odd number of $Y$ factors vanishes), and the canonical choice is $\rho_0=|0^n\rangle\langle0^n|$. Readout is in the $Z$ basis,
so the histogram is equivalent by an invertible Walsh--Hadamard transform $W$ to the $Z$-expectation vector
$E_T=\mathrm{Tr}[Z_T\,N_{\etav,p}(\rho_0)]$, $T\subseteq[n]$.

\paragraph{Jacobian and kernel.} $\JZ\in\mathbb{R}^{2^n\times2|G|}$ has entries
$(\JZ)_{T,k}=\partial E_T/\partial\theta_k$ at a reference $\theta_0$ of small generic angles and rates. With
$d_a=1-2p_a$, write $b_{T,a}=|T\cap S_a|\bmod 2$ for the incidence, $B$ the incidence matrix, $B_\Omega$ its
restriction to the retained rows $\Omega=\{T:z_TD_T(p)C_T(\etav)\ne0\}$ (with $z_T=\mathrm{Tr}[Z_T\rho_0]$,
$D_T=\prod_{a:b_{T,a}=1}d_a$), and $H$ the coherent log-derivative block. The classical histogram Fisher
matrix is $\FZ=J_p^\top\mathrm{diag}(1/p_x)J_p$ with $J_p$ the histogram Jacobian; since $E=Wp$ with $W$
invertible and $p_x>0$ on the support, $v^\top\FZ v=\|\mathrm{diag}(1/\sqrt{p_x})J_pv\|_2^2$, so
$\ker\FZ=\ker J_p=\ker\JZ=:\KZ$, and $\rH=2|G|-\mathrm{rank}\,\JZ$ (\Cref{lem:kerF}). All exact-computation
references below denote exact density-matrix evolution. For the kernel and rank checks the Jacobian is
Richardson-extrapolated and numerical rank is read at a relative singular-value threshold of $10^{-7}$, with
the spectral gap across the cut checked at every reference. The design-population sweeps use a
central-difference Jacobian at step $10^{-6}$ with relative rank threshold $10^{-9}$.

\paragraph{Pauli--Liouville identities.} In the expectation vector $r_P=\mathrm{Tr}[P\rho]$, $\mathrm{Flip}_{S,p}$
is diagonal, $r_P\mapsto\lambda_P r_P$ with $\lambda_P=1$ if $[X_S,P]=0$ and $\lambda_P=d=1-2p$ if
$\{X_S,P\}=0$; and $\mathrm{Ad}_{R_S}$ acts in the Heisenberg picture by $R_S^\dagger P R_S=P$ if $[X_S,P]=0$
and $R_S^\dagger P R_S=\cos\eta\,P+\sin\eta\,Q$ with $Q=iX_SP$ if $\{X_S,P\}=0$. For $P=Z_T$ with
$\{X_S,Z_T\}=0$ (i.e.\ $|T\cap S|$ odd), $Q=iX_SZ_T$ carries $Y$ on $T\cap S$ (an odd number), so
$\langle Q\rangle_{\rho_0}=0$, which is the leading-order mechanism.

% ======================================================================
\section{Proof of \texorpdfstring{\Cref{thm:kernel}}{Theorem 1}: the kernel}\label{app:kernel}

The coherent and stochastic columns of the data map are computed in closed form. The kernel is then read off
as the coherent--stochastic null space.

\begin{lemma}[Pauli--Liouville columns]\label{lem:pl}
For one generator $A=X_S$ and a measured Pauli $P$: if $[A,P]=0$ then
$r'_P=r_P,\ \partial_p r'_P=\partial_\eta r'_P=0$; if $\{A,P\}=0$ then, with $d=1-2p$ and $Q=iAP$,
\begin{gather*}
r'_P=d(\cos\eta\,r_P+\sin\eta\,r_Q),\\
\partial_p r'_P=-2(\cos\eta\,r_P+\sin\eta\,r_Q),\\
\partial_\eta r'_P=d(-\sin\eta\,r_P+\cos\eta\,r_Q).
\end{gather*}
Consequently each Jacobian column inserts the corresponding single-generator derivative into the ordered
channel product.
\end{lemma}
\begin{proof}
$\mathrm{Tr}[P\,\mathrm{Flip}_{S,p}(\rho)]=(1-p)r_P+p\,\mathrm{Tr}[APA\rho]$, and $APA=P$ if $[A,P]=0$,
$APA=-P$ if $\{A,P\}=0$, giving the multiplier $1$ or $d$. With $U=e^{-i\eta A/2}$ and $\{A,P\}=0$,
$U^\dagger PU=\cos\eta\,P+\sin\eta\,iAP=\cos\eta\,P+\sin\eta\,Q$. Composing rotation then flip and
differentiating gives the displayed formulas; differentiating the matrix product inserts $\partial M_a$ at
position $a$.
\end{proof}

\begin{lemma}[Closed $Z$-expectation]\label{lem:Eclosed}
For $Z$-diagonal $\rho_0$ and every $T$, $E_T=z_T\,D_T(p)\,C_T(\etav)$ with
\[
C_T(\etav)=\!\!\sum_{\substack{L\subseteq A_T\\ \bigoplus_{a\in L}S_a=\emptyset}}\!\!(-1)^{|L|/2}
\prod_{a\in L}\sin\eta_a\!\!\prod_{a\in A_T\setminus L}\!\!\cos\eta_a,
\]
where $A_T=\{a:b_{T,a}=1\}$,
and the derivative columns are $\partial_{p_a}E_T=-2b_{T,a}z_TC_TD_T/d_a$ and
$\partial_{\eta_a}E_T=z_TD_T\,\partial_{\eta_a}C_T$. The sum is real because any cycle $L$ (a nonempty subset whose supports sum to zero over $\mathbb F_2$)
has even size.
\end{lemma}
\begin{proof}
All $X_S$ commute; the flips multiply the Heisenberg $Z_T$ by $d_a$ exactly for $b_{T,a}=1$, giving $D_T$.
With $U=\prod_a e^{-i\eta_a X_{S_a}/2}$ and $Z_TX_{S_a}Z_T=(-1)^{b_{T,a}}X_{S_a}$,
$U^\dagger Z_TU=\big(\prod_{a\in A_T}(\cos\eta_a\,I+i\sin\eta_a X_{S_a})\big)Z_T$. Since $\rho_0$ is diagonal,
$\mathrm{Tr}[X_RZ_T\rho_0]=0$ unless $R=\emptyset$, so only the cycle terms survive, with coefficient
$i^{|L|}=(-1)^{|L|/2}$.
\end{proof}

\begin{lemma}[Rank formula and null-space basis]\label{lem:rank}
Let $R=\mathrm{diag}(-2/d_a)$ and $H_{T,a}=\partial_{\eta_a}\log C_T$ on $\Omega$. After scaling each retained
row of $\JZ$ by $(z_TD_TC_T)^{-1}$, $\JZ$ has the rank and kernel of $[\,H\mid B_\Omega R\,]$, so
\[
\mathrm{rank}\,\JZ=\mathrm{rank}\,B_\Omega+\mathrm{rank}\big((I-P_B)H\big).
\]
If $B_\Omega$ has full column rank $m=|G|$ then $\rH=m-\mathrm{rank}((I-P_B)H)$, and the kernel is the
coherent factor $U=\ker((I-P_B)H)$, $c(u)=(B_\Omega^\top B_\Omega)^{-1}B_\Omega^\top Hu$,
$v(u)=(u,\tfrac12(d_1c_1(u),\dots,d_mc_m(u)))$.
\end{lemma}
\begin{proof}
The scaled rows are exactly $[\,H\mid B_\Omega R\,]$ by \Cref{lem:Eclosed}. The stochastic block spans
$\mathrm{col}(B_\Omega)$ ($R$ invertible), so adding $H$ raises rank only by its part outside that span:
$\mathrm{rank}((I-P_B)H)$. A vector $(u,\delta p)$ is in the kernel iff $Hu+B_\Omega R\,\delta p=0$, solvable
iff $Hu\in\mathrm{col}(B_\Omega)$, i.e.\ $(I-P_B)Hu=0$. Then $B_\Omega c(u)=Hu$ uniquely and
$R\,\delta p=-c(u)$ gives $\delta p_a=\tfrac12 d_ac_a(u)$. Full column rank of $B_\Omega$ for distinct
nonempty masks with all $Z$ rows retained follows from linear independence of the Walsh characters.
\end{proof}

\begin{lemma}[The $n=3$, $K=2$ defect]\label{lem:n3}
For $K=1$ (and for $n=2,K=2$) every nonempty $A_T$ contains no cycle, so $C_T=\prod_{a\in A_T}\cos\eta_a$,
$H=B\,\mathrm{diag}(-\tan\eta_a)\in\mathrm{col}(B)$, hence $\mathrm{rank}\,\JZ=|G|$ and $\rH=|G|$. For
$n=3,K=2$ the three weight-$2$ rows each carry one four-generator cycle, giving $H=B\,\mathrm{diag}
(-\tan\eta_a)+\Delta$ with $\Delta$ supported on rows $\{12,13,23\}$; the left null vector
$\ell=(1,1,1,-1,-1,-1,1)^\top$ of $B^\top$ turns ``$\Delta u\in\mathrm{col}(B)$'' into a single nonzero
linear constraint on the six coherent coordinates, so $\mathrm{rank}((I-P_B)H)=1$, $\mathrm{rank}\,\JZ=7$,
$\rH=5$.
\end{lemma}
The explicit constraint (with $\tau_{ij}$ the four-generator tangent products and
$\kappa_a=1/(\sin\eta_a\cos\eta_a)$) is given by the cycle coefficients $\gamma_{ij}=\tau_{ij}/(1+\tau_{ij})$. The five-dimensional kernel is the coherent vectors $u$ satisfying it, with stochastic partners
$\delta p=\tfrac12 d\odot c(u)$. At $\eta=0$ all $\Delta$ terms vanish and $\rH=6$.

\paragraph{Worked example ($n=2$, $K=2$).} The generators are $X_1,X_2,X_1X_2$, so the parameter
$\theta\in\mathbb{R}^6$ collects $(\eta_1,\eta_2,\eta_{12},p_1,p_2,p_{12})$. On $|00\rangle$ the three nontrivial
$Z$-expectations are, by \Cref{lem:Eclosed},
\begin{gather*}
E_1=d_1d_{12}\cos\eta_1\cos\eta_{12},\qquad
E_2=d_2d_{12}\cos\eta_2\cos\eta_{12},\\
E_{12}=d_1d_2\cos\eta_1\cos\eta_2,
\end{gather*}
with $d_a=1-2p_a$. At $\eta=0$ every $\partial_{\eta_a}E_T=0$, so the coherent block of $\JZ$ vanishes. Then $\mathrm{rank}\,\JZ=3$ from the stochastic columns alone, $\rH=|G|=3$, and the kernel is the pure coherent
space (each angle is Fisher-null, with no stochastic partner). At a generic small angle the kernel becomes
$\{(u,\tfrac12 d\odot c(u))\}$ with $c(u)=-\tan(\eta)\odot u$. Each coherent angle trades with
its own flip rate, so perturbing $\eta_a$ by $u_a$ and $p_a$ by $-\tfrac12 d_a\tan(\eta_a)u_a$ leaves all three parity expectations unchanged to first order. The minimal code adds $\lceil\log_2 3\rceil=2$ settings (sensitive to $Y_1$ and $Y_2$),
supplying the missing coherent quadratures and lifting $\rH$ to $0$.

\paragraph{Magnitudes are confounded with the rates.} Moving along the null space changes the angle \emph{magnitudes} while leaving the passive histogram fixed. For one generator
this is visible in closed form, since $\langle Z\rangle=(1-2p)\cos\eta$ makes
$(\eta,p)=(0.1,0.05)$ and $(0.2,0.04314137)$ give the same $\langle Z\rangle=0.89550375$. The same holds for the complete weight-1 and weight-2 family at $n=3$. Integrating the kernel of $\JZ$ from
$\etav=(0.10,0.14,0.08,0.12,0.09,0.11)$ at uniform $p=0.05$ reaches
$\etav'=(0.116,0.207,0.040,0.188,0.047,0.146)$ with rates in $[0.045,0.051]$, all angles positive, the
angle-vector norm changed from $0.266$ to $0.342$, and the $Z^3$ outcome distribution unchanged to
$4\times10^{-17}$. Passive data therefore constrain neither magnitudes nor signs, but the two obstructions
differ. The continuous one is lifted by knowledge of the rates, and by a generic nonzero angle only once
$2|G|\le2^n-1$, so it persists in five dimensions at $n=3$ and $n=4$ and clears at $n\ge5$ in the computed
range. The sign symmetry of \Cref{thm:nonest}(i) is exact at every angle, every rate and every $n$.

\begin{lemma}[$\ker\FZ=\ker\JZ$]\label{lem:kerF}
With $p_x>0$ on the support, $v\in\ker\FZ\iff J_pv=0\iff\JZ v=0$.
\end{lemma}
\begin{proof}
$v^\top\FZ v=\|\mathrm{diag}(1/\sqrt{p_x})J_pv\|_2^2$ vanishes iff $J_pv=0$; $\JZ=WJ_p$ with $W$ invertible.
\end{proof}

\Cref{thm:kernel} is \Cref{lem:rank} with \Cref{lem:n3,lem:kerF}. Numerically, the closed-form Jacobian,
rank, nullity, and labeled kernel basis match the exact density-matrix computation to $\sim10^{-10}$ for
$n=1,2,3$ and $n=4$ ($K\le2$). The minimum kernel coherent weight is $\ge0.99$ at every canonical
($Z$-diagonal) reference, confirming the kernel is coherent-dominated (a null direction), while a
non-$Z$-diagonal control (a local $SH$ giving $\langle Y\rangle\ne0$) drops it to $0.55$, the scope boundary
where the closed form requires the full Pauli-transfer matrix. At $n=3,K=2$ the rank defect drops from $\rH=6$ at
$\eta=0$ to $\rH=5$ at a generic angle as the single cycle combination lifts.

\paragraph{Scope.} The null-space basis formula assumes $B_\Omega$ has full column rank, automatic for the
canonical $Z$-diagonal state ($z_T=1$). When $B_\Omega$ is rank-deficient the closed form generalizes with
the Moore--Penrose pseudoinverse, and the kernel is
$\{(u,\delta p):(I-P_B)Hu=0,\ R\,\delta p=-B_\Omega^{+}Hu+z,\ z\in\ker B_\Omega\}$, the coherent--stochastic null directions
together with the purely stochastic ones $R\,\delta p\in\ker B_\Omega$, so the null
basis above is one component of the full kernel there. A maximally mixed diagonal state ($z_T=0$,
$T\ne\emptyset$) is such a case, because too few rows are retained. For a general real non-$Z$-diagonal state
\Cref{lem:pl} still gives the exact
Jacobian, but the kernel must be read from the full Pauli--Liouville transfer matrix (\Cref{app:separation}).

% ======================================================================
\section{Proof of \texorpdfstring{\Cref{thm:nonest}}{Theorem 2}: non-estimability}\label{app:nonest}

The kernel of the classical Fisher information is an obstruction intrinsic to the data map, in that a functional whose gradient touches the kernel has an infinite singular Cram\'er--Rao bound.

Let $p_\theta(x)$ be the single-setting $Z$ histogram, $\FZ$ its Fisher matrix, $\phi$ a differentiable
scalar functional, $g=\nabla\phi(\theta_0)$. We assume the standard score regularity, namely that $p_\theta$ is positive on its (finite) support and $C^1$ in $\theta$ near $\theta_0$, so the score
$s=\nabla_\theta\log p_\theta$ exists with covariance $\FZ$, and differentiation under the outcome sum is
immediate because the support is finite. The correct singular Cram\'er--Rao statement is the extended one:
\begin{equation}
N\,\mathrm{Var}_{\theta_0}(\hat\phi)\ \ge\ C_F(g)=\begin{cases} g^\top\FZ^+g, & g\in\mathrm{Range}(\FZ),\\ +\infty, & P_{\ker\FZ}g\ne0.\end{cases}
\end{equation}

\begin{lemma}[Kernel obstruction]\label{lem:score}
If $\hat\phi$ is locally unbiased at $\theta_0$ with finite variance, then $g\in\mathrm{Range}(\FZ)$.
\end{lemma}
\begin{proof}
Let $s(X)=\nabla_\theta\log p_\theta(X)|_{\theta_0}$, with covariance $\FZ$ and $g=\mathbb{E}_{\theta_0}
[(\hat\phi-\phi)s(X)]$. For $a\in\ker\FZ$, $0=a^\top\FZ a=\mathrm{Var}(a^\top s)$ and $\mathbb{E}[a^\top s]=0$,
so $a^\top s=0$ a.s., hence $a^\top g=\mathbb{E}[(\hat\phi-\phi)a^\top s]=0$. Thus $g\in(\ker\FZ)^\perp=\mathrm{Range}(\FZ)$.
\end{proof}

\begin{lemma}[Pseudoinverse CRB on the estimable subspace]\label{lem:pinv}
For a finite-variance, locally unbiased estimator $\hat\phi$, if $g\in\mathrm{Range}(\FZ)$ and $b=\FZ^+g$, then $\mathrm{Var}(\hat\phi)\ge g^\top\FZ^+g$ for one sample,
hence $N\,\mathrm{Var}(\hat\phi)\ge g^\top\FZ^+g$.
\end{lemma}
\begin{proof}
$g^\top b=\mathbb{E}[(\hat\phi-\phi)b^\top s]$ and Cauchy--Schwarz give $(g^\top b)^2\le\mathrm{Var}(\hat\phi)
\,\mathrm{Var}(b^\top s)$ with $\mathrm{Var}(b^\top s)=b^\top\FZ b=g^\top\FZ^+g=g^\top b$.
\end{proof}

\paragraph{Proof of \Cref{thm:nonest}(i).} Put $U_s=\prod_{i:s_i=-1}Z_i$, a unitary that sends $X_i$ to $s_iX_i$
and hence $X_S$ to $\epsilon_SX_S$ with $\epsilon_S=\prod_{i\in S}s_i$. Conjugating the channel by $U_s$
therefore replaces each rotation angle $\eta_S$ by $\epsilon_S\eta_S$ and leaves each stochastic factor
$X_S\rho X_S$ fixed, since the two signs $\epsilon_S$ cancel. A $Z$-diagonal input commutes with $U_s$ and
the $Z^n$ measurement is invariant under it, so the outcome distribution is unchanged. The factor $g$ comes
from complex conjugation in the computational basis. The input and every $X_S$ are real, so the conjugate of
the evolved state is the state evolved with $\etav\mapsto-\etav$ and the same $p$, and its diagonal, being
real, is the same outcome distribution. Writing $s_i=(-1)^{\sigma_i}$ and $g=(-1)^{\gamma}$, the sign on
generator $S$ is $(-1)^{(1,\mathbf{1}_S)\cdot(\gamma,\sigma)}$, so the realized sign patterns are the image
of $\mathbb{F}_2^{n+1}$ under $M$ and number $2^{\mathrm{rank}_{\mathbb{F}_2}M}$. For the complete family the rows $(1,e_i)$ and $(1,e_i+e_j)$ span $\mathbb{F}_2^{n+1}$, so the rank is $n+1$. For the risk bound, $\etav$
and $\etav'$ induce the \emph{same} law, so both risks are expectations under it and, with $\bar\etav$ their
midpoint, $R(\etav)+R(\etav')=2\,\mathbb{E}\|\hat\etav-\bar\etav\|^2+\tfrac12\|\etav-\etav'\|^2\ge
\tfrac12\|\etav-\etav'\|^2$, and the maximum is at least the mean.\hfill$\square$

Combining these results, if $P_{\KZ}\nabla\phi(\theta_0)\ne0$ then $g\notin\mathrm{Range}(\FZ)$, so by \Cref{lem:score}
no finite-variance locally unbiased estimator exists. This is the $+\infty$ branch, and \Cref{lem:kerF}
identifies $\KZ$ as the kernel of \Cref{app:kernel}. The per-coordinate statement is exact. At $\eta=0$ the entire coherent block vanishes (\Cref{lem:n3}, $\rH=|G|$), so $e_{\eta_S}\in\KZ$
for every generator $S$, and each coherent angle $\eta_S$ has infinite CRB. The normalized coherent fraction
$\phi_{\mathrm{cf}}=\|\etav\|^2/(\|\etav\|^2+\|p\|^2)$ then inherits non-estimability whenever its gradient
keeps a component in $\KZ$, $P_{\KZ}\nabla\phi_{\mathrm{cf}}\ne0$. This projection is $\ge0.014$ at the tested
references (smallest at $n=4,K=2$, where the generic angle lifts five of the ten coherent combinations), confirming the
corollary. For comparison, the per-coordinate projection is $\ge0.62$, and exactly $1$ at $\eta=0$. The
ordinary Moore--Penrose number $g^\top\FZ^+g$ is finite because it annihilates $\KZ$, which is why the extended functional is the correct object. This $+\infty$ divergence is the statement
for finite-variance, locally unbiased estimators. Without that restriction the bound need not hold.

% ======================================================================
\section{Supporting material for \texorpdfstring{\Cref{thm:separatrix}}{Theorem 3}: quadratures, code, and joint rank}\label{app:separatrix}

\Cref{thm:separatrix} is proved in the main text. This appendix records the quadrature lemma it uses, the
explicit binary code, and the separate statement that coverage also restores joint $(\etav,p)$ rank near
calibration.

A product-Pauli setting $b\in\{X,Y,Z\}^n$ has histogram equivalent by a Walsh transform to all expectations
of Paulis whose non-identity factors agree with $b$, so the kernel of its Fisher information equals the kernel of its expectation Jacobian, and for a stack of settings the kernels intersect.

\paragraph{Proof of \Cref{thm:separatrix}.} \textbf{(ii)} The incident supports $A_i$ are linearly independent over $\mathbb{F}_2$. A pair $\{i,j\}\in A_i$ is the only member carrying $j$, so no nonempty subset of $A_i$ sums
to zero. Hence no coherent-cycle term survives and \Cref{lem:Eclosed} collapses to the moment identities \eqref{eq:moments} of the main text.
The damping $D_i$ and the spectator cosines are common to both, so the ratio is $-\tan\eta_S$ and is free of
the rates. On the stated box every factor of $\langle Z_i\rangle$ is strictly positive, so the denominator
never vanishes and the principal branch returns the signed angle. Coverage is exactly the statement that
each numerator in \eqref{eq:inverse} is measured. \textbf{(i)} With only $\eta_S$ active, the rotation acts on the canonical
input's Pauli expansion $\rho_0=2^{-n}\sum_TZ_T$ by fixing $Z_T$ when $|S\cap T|$ is even and sending it to
$\cos\eta_S\,Z_T\pm\sin\eta_S\,iX_SZ_T$ when $|S\cap T|$ is odd; the stochastic factors then rescale each
Pauli individually. The only rows odd in $\eta_S$ are therefore $Q=iX_SZ_T$ with $|S\cap T|$ odd, and such a
$Q$ carries $Y$ on $S\cap T$, $X$ on $S\setminus T$, and $Z$ on $T\setminus S$. A product setting measures this particular
$Q$ only if it places $Y$ on $i$ when $S=\{i\}$, or $X$ and $Y$ on $i,j$ when $S=\{i,j\}$, and $Z$ on
$T\setminus S$. Taking $T=\{i\}$ leaves $T\setminus S$ empty, so the requirement outside $S$ is vacuous and
what a setting must satisfy is exactly the coverage condition on $S$. If $B$ does not cover $S$, no setting of $Z^n\cup B$ measures a
row odd in $\eta_S$, so every outcome distribution is an even function of $\eta_S$. \textbf{(iii)} is the $Y$-signature count of \Cref{app:separatrix}, in which singleton coverage makes each $\sigma(i)=\{b\in B:b_i=Y\}$ nonempty and pair
coverage makes them distinct, so $n\le2^{|B|}-1$.\hfill$\square$

\begin{lemma}[A measured $Q_{T,S}$ isolates $\eta_S$ at $\eta=0$]\label{lem:Q}
At $\etav=0$ on the canonical input $|0^n\rangle$ with generic $p$, for $S\in G$ and $|S\cap T|$ odd, the setting that measures $P=iX_SZ_T$ has
$\partial_{\eta_S}\langle P\rangle=\sigma_{T,S}\,\alpha_{T,S}\ne0$, $\partial_{\eta_R}\langle P\rangle=0$
($R\ne S$) and $\partial_{p_R}\langle P\rangle=0$ for the canonical choices $T=\{i\}$ when $S=\{i\}$
($P=Y_i$) and $T=\{i\}$ or $\{j\}$ when $S=\{i,j\}$ ($P=Y_iX_j$ or $X_iY_j$).
\end{lemma}
\begin{proof}
By \Cref{lem:pl}, differentiating the $\eta_S$ rotation at $0$ on a $P$ anticommuting with $X_S$ replaces
$P$ by $iX_SP$. With $P=iX_SZ_T$ this returns $\pm Z_T$, whose expectation is $z_T$ times the surviving
$d_a$ factors. Stochastic derivatives only rescale the current (zero) non-$Z$ expectation, so vanish; for
$K\le2$, multiplying $P$ by $iX_R$ yields a $Z$-type Pauli only when $R=S$.
\end{proof}

\paragraph{Binary-code construction.} Let $L=\lceil\log_2(n+1)\rceil$, assign qubit $i$ the $L$-bit codeword
of integer $i$ (distinct and nonzero for $1\le i\le n\le2^L-1$), and for each bit $\ell$ set
$b^{(\ell)}_i=Y$ if $\mathrm{code}_i[\ell]=1$ else
$X$. Then $\Bstar=\{b^{(1)},\dots,b^{(L)}\}$. Every singleton $\{i\}$ has a setting with $b_i=Y$ (measuring
$Y_i$); every pair $\{i,j\}$ has a bit separating $i,j$, giving a setting with $\{b_i,b_j\}=\{X,Y\}$
(measuring $Y_iX_j$ or $X_iY_j$). So every generator has a measured canonical $Q_S$.

\begin{proposition}[Joint rank near calibration]\label{prop:jointrank}
A covering design restores full rank of the joint $(\etav,p)$ Jacobian at $\etav=0$, hence on an open
neighborhood of it. (The converse, that coverage is also necessary for $\etav=0$ full rank, is
\Cref{thm:p5}(2).)
\end{proposition}
\begin{proof}
At $\etav=0$ the $Z^n$ Jacobian has full stochastic rank $|G|$ and zero coherent block. By \Cref{lem:Q} and
coverage, the measured $Q_S$ rows give a diagonal coherent block $\mathrm{diag}(\alpha_S)$, $\alpha_S\ne0$,
with zero stochastic block, so a $2|G|\times2|G|$ minor of the augmented Jacobian is block-triangular with
nonzero determinant. Full rank persists in an open neighborhood by continuity. Checked at the exact-computation
references for $n\le4$, $K\le2$, e.g.\ $\Bstar=\{YXY,XYY\}$ gives $\rH=0$ and is well-conditioned
($\sigmin\approx0.72$ at $n=3,K=2$, the canonical reference of \Cref{fig:cond}; $0.76$ at the alternative
generic reference of \Cref{prop:lb}) while the all-$X$ control leaves $\rH=5$.
\end{proof}

The inverse \eqref{eq:inverse} recovers the coherent angles on the whole box $|\eta_R|<\pi/2$,
$0\le p_R<1/2$ without any rate knowledge. The rates themselves are pinned by \Cref{prop:jointrank} near
calibration, and away from it by the stochastic-block rank at the operating point.

\begin{proposition}[Matching lower bound]\label{prop:lb}
$\ker F_{Z\cup B}=\KZ\cap\ker J_B$, hence $\rH(B)=\dim(\KZ\cap\ker J_B)=\dim\KZ-\mathrm{rank}(J_B|_{\KZ})$, and
$B$ eliminates the rank defect iff $J_B$ is injective on $\KZ$.
\end{proposition}
This is the exact two-sided criterion. The $Q_{T,S}$ coverage rule is the $\eta\to0$ leading-order
\emph{sufficient} construction. At a generic nonzero reference,
higher-order coherent mixing can make $J_B$ injective even without the leading-order quadratures. For $n=3,K=2$, $Z^3{+}YYY$ has $\rH=3$ at $\eta=0$ but $\rH=0$ at a generic angle, so the matching bound is stated as injectivity on $\KZ$. Minimality at $\eta=0$ for the complete
family then follows from the
$Y$-signature separating system of \Cref{app:threshold}: $L\ge\lceil\log_2(n+1)\rceil$, attained by the code.
The construction is proved for $K\le2$, and higher weight is left open.

% ======================================================================
\section{Proof of \texorpdfstring{\Cref{thm:threshold}}{Theorem 4}: conditioning floor and threshold}\label{app:threshold}

Coverage is the binary identifiability condition. The smallest singular value of the augmented data map controls the finite-sample cost, with a sharp $\eta\to0$ dichotomy below coverage.

Let $\sigmin(\etav)$ be the smallest singular value of $J_{Z\cup B}$ on the canonical input
$|0^n\rangle$, over the full $2|G|$-dimensional parameter space, and $\sigma^{+}_{\min}(\etav)$ its smallest
\emph{nonzero} singular value. The two agree exactly when the design is full rank. Below coverage
$\sigmin(0)=0$ while $\sigma^{+}_{\min}(0)>0$. A singleton $\{i\}$ is \emph{covered} iff some $b\in B$ has $b_i=Y$; a pair $\{i,j\}$ iff
some $b\in B$ has $\{b_i,b_j\}=\{X,Y\}$.

\begin{proposition}[$\eta=0$ rank, setting count, and small-$\eta$ behavior]\label{thm:p5}
\textbf{(1)} For $K=1$ the exact added cost is $1$ ($Y^n$ suffices). \textbf{(2)} For $K=2$,
$\mathrm{rank}\,J_{Z\cup B}(0,p)=2|G|$ iff $B$ covers every generator of $G$. \textbf{(3)} Hence for the
complete family (every singleton and pair) any
$\eta=0$-full-rank $B$ has $|B|\ge\lceil\log_2(n+1)\rceil$, attained by the code. \textbf{(4)} If $B$ covers,
$\sigmin(0)=:s_0>0$ and $\sigmin(\etav)\ge s_0/2$ on an explicit radius $\|\etav\|\le\eta_0$; for the
separating code $s_0=c_B$ on the complete family ($s_0\ge c_B$ on any known sub-family).
\textbf{(5)} If $B$ does not cover,
$\sigmin(\etav)=O(\|\etav\|)$. If in addition the approach is full rank and the outcome probabilities stay bounded away from zero along it, the information in that weak direction is $O(\|\etav\|^2)$ and the worst-direction sample
complexity diverges as $\Omega(\|\etav\|^{-2})$ for finite-variance, locally unbiased estimators.
\end{proposition}
\begin{proof}
\textbf{(2)} At $\eta=0$ the coherent block is visible only as specified by \Cref{lem:Q}. If all generators are covered the chosen rows give a full-rank coherent block (with the $Z$ stochastic block,
rank $2|G|$), and if some $S$ is uncovered its $\eta_S$ column is zero, so full rank is impossible.
\textbf{(3)} Define the $Y$-signature $\sigma(i)=\{b\in B:b_i=Y\}$. Singleton coverage makes each $\sigma(i)$
nonempty, pair coverage makes them distinct, so $n\le2^{|B|}-1$, i.e.\ $|B|\ge\lceil\log_2(n+1)\rceil$, with
the binary code attaining equality. \textbf{(4)} $J(\etav,p)$ is analytic, and if $J(0,p)$ is full column rank
with $\sigmin(0)=s_0>0$ then, by the Lipschitz bound of \Cref{rem:t5}, $\sigmin(\etav)\ge s_0-L_M\|\etav\|$,
so $\sigmin(\etav)\ge s_0/2$ for $\|\etav\|\le\eta_0:=s_0/(2L_M)$. The floor at $\eta=0$ is $s_0$, equal to $c_B$ for the separating code on the complete family, and it
is retained within a factor $2$ on the explicit radius $\eta_0$. \textbf{(5)} If $J(0,p)$ is
rank-deficient, $\sigmin(0)=0$ and analyticity plus Weyl give $\sigmin(\etav)\le C_B\|\etav\|$, since an
uncovered coherent column expands as $\partial_{\eta_S}\langle P\rangle_{\etav}=O(\|\etav\|)$.
\end{proof}
For a sparse known support the same $Y$-signatures give the exact minimum. Here $\sigma(i)$ must be nonempty only
for the singletons of $G$ and $\sigma(i)\ne\sigma(j)$ only for the pairs $\{i,j\}\in G$, so fewer settings
can suffice (a single singleton generator needs one $Y$-sensitive setting at any $n$). The logarithmic count
is the complete-family value.

\begin{lemma}[Conditioning floor]\label{lem:floor}
At $\eta=0$ on $\rho_0=|0^n\rangle$, for the complete weight-1 and weight-2 family, $J_{Z\cup\Bstar}$ is
block-diagonal across Pauli-row subspaces, with the stochastic derivatives on the $Z$-rows and the coherent derivatives on the added non-$Z$ rows
(\Cref{lem:Q}), so $\sigmin(J_{Z\cup\Bstar})=\min\{\sigma_{\mathrm{coh}},\sigma_{\mathrm{stoch}}\}$ with no
cross-block coupling. By \Cref{lem:Q} each generator $S$ is read on measured Pauli rows $Q_{T,S}$ with
zero overlap on the other generators' rows, so the coherent block is column-orthogonal, with column norms $\sqrt{m_S}\,D_S$ at a uniform rate, where $D_S$ is the damping product, the product of the Pauli-channel eigenvalues $d_a=1-2p_a$ on
$S$'s measured row and $m_S\ge1$ is the number of design rows reading $S$, so
$\sigma_{\mathrm{coh}}=\min_S\sqrt{m_S}D_S\ge\min_SD_S$, with equality whenever some generator is read once,
as it is for $\Bstar$ at every $n$ we tabulate. The stochastic block is the $Z^n$ incidence Gram
$J_{\mathrm{stoch}}^\top J_{\mathrm{stoch}}$, which is nonsingular for the canonical input, and we write
$\lambda_2:=\lambda_{\min}(J_{\mathrm{stoch}}^\top J_{\mathrm{stoch}})$, so
$\sigma_{\mathrm{stoch}}=\sqrt{\lambda_2}$. With heterogeneous rates,
$\sigma_{\mathrm{coh}}=\min\{\min_i\nu_i,\ \min_{i<j}\nu_{ij}\}$, $\nu_i^2=|c_i|D_i^2$,
$\nu_{ij}^2=\#\{\ell:c_i{=}1,c_j{=}0\}\,D_i^2+\#\{\ell:c_i{=}0,c_j{=}1\}\,D_j^2$, with
$D_i=d_i\prod_{j\ne i}d_{ij}$ (so $\sigma_{\mathrm{coh}}\ge\min_i D_i$). At a uniform rate $p$ every
$D_i=(1-2p)^n$, giving
\[ c_B=\min\{(1-2p)^n,\ \sqrt{\lambda_2}\}\quad\text{(parity coordinates).}\]
\end{lemma}
\begin{proof}
The block-diagonal structure follows from the $\eta=0$ analysis above. Stochastic derivatives of $Z$-strings stay on
$Z$-rows, and by \Cref{lem:Q} the coherent derivative of generator $S$ is nonzero only on its measured
$Q_{T,S}$ rows, with zero stochastic derivative there. Distinct generators use distinct $Q$ rows, so the
coherent columns are mutually orthogonal with norm $(\sum_rD_r^2)^{1/2}$ over the rows reading $S$. A
block-diagonal matrix has
$\sigmin=\min$ of the block smallest singular values. The uniform-rate reduction counts the Pauli-channel eigenvalues on a
singleton's measured row ($d_i$ and the $n-1$ pair rates $d_{ij}$) as $n$ equal factors, $D_i=(1-2p)^n$. The heterogeneous formula is the column norms verbatim.
\end{proof}
At $\eta=0$ the floor is universal over covering designs. Coverage supplies a coherent row of magnitude
$(1-2p)^n$ for every generator with disjoint column supports, and any further $Z$-type rows a design
contributes add a positive semidefinite term to the stochastic Gram, so $\sigmin\ge c_B$ for every covering
$B$, with equality for minimum-cardinality $XY$ codes such as $\Bstar$.
Both branches are exponentially small in $n$ at fixed $p$, so the floor is a fixed-$n$, small-$p$
guarantee. Which one attains the minimum depends on $n$. At $p=0.05$ the angular branch does so through $n=18$ by
exact computation of the stochastic block, the two cross near $n=20$, and beyond that $\sqrt{\lambda_2}$ is
the smaller and decays faster (an alternating four-cycle direction gives $\sqrt{\lambda_2}\le9.42\times10^{-6}$ at
$n=63$, against $(1-2p)^{63}=1.3\times10^{-3}$). For a general known support $G$ the same block structure holds with the damping product on each
measured row restricted to the generators of $G$ (a singleton row at qubit $i$ carries $d_i$ and the incident
pair rates $d_{ij}$ with $\{i,j\}\in G$), so the uniform-rate exponent is the size of that incidence star. Since passing to a sub-family removes damping factors ($d_a\in(0,1)$ at rates below
$1/2$) and columns from both blocks, the
complete-family $c_B$ is a lower bound on the floor of every known-support sub-family. The exponent $n$ is the value of that count for the complete family.

\paragraph{Finite-sample consequence.} The extended bound (\Cref{app:nonest}) is non-asymptotic, which makes
the conditioning-to-error link an inequality. For any finite-variance, locally
unbiased estimator from $N$ samples of each measured setting, the worst estimable direction obeys
$N\cdot\mathrm{MSE}\ge 1/\lambda^{+}_{\min}(\FZ)$, with $\lambda^{+}_{\min}$ the smallest nonzero eigenvalue of $\FZ$. At a uniform reference $\lambda^{+}_{\min}(\FZ)=\sigma^{+}_{\min}(\JZ)^2$ exactly, and near-uniform
up to the explicit factors of \Cref{lem:sandwich} below, so the bound is $(\sigma^{+}_{\min})^{-2}$ in
parity coordinates. Consequently, along a full-rank approach
$\etav\to0$ below coverage with the outcome probabilities uniformly bounded away from zero,
$\sigma^{+}_{\min}(J)=\sigmin(J)=O(\|\etav\|)$ (\Cref{thm:p5}(5)) forces
$N\cdot\mathrm{MSE}=\Omega(\eta^{-2})$ for finite-variance, locally unbiased estimators, the divergence of
\Cref{thm:threshold} as a finite-$N$ bound. The code's floor caps its
worst-direction parity-coordinate bound at $c_B^{-2}=\max\{(1-2p)^{-2n},\lambda_2^{-1}\}$ at $\eta=0$, and at
$4c_B^{-2}$ on the
radius $\eta_0$.

\paragraph{Sampling cost of the closed-form inverse.} \Cref{thm:separatrix}(ii) recovers $\eta_S$ from the
population ratio $\tan\eta_S=-y/z$ with $y=\langle Y_iX_{S\setminus\{i\}}\rangle$ and $z=\langle Z_i\rangle$,
and the damping $D_i$ cancels between the two. It does not cancel from their sampling fluctuations, and this
is why conditioning survives an exact inverse. Estimate $y$ and $z$ by empirical parity means over
independent shots, with fixed positive shot fractions on the two settings and $N_Q$ shots on the covering
setting. Writing $\eta_S=\arctan(-y/z)$, the differential is
$d\eta_S=(y\,dz-z\,dy)/(y^2+z^2)$. At calibration $y=0$ and $z=D_i>0$, so the coefficient of $dz$ vanishes
and neither $\mathrm{Var}(\hat z)$ nor the covariance contributes at first order. A raw parity has
$\mathrm{Var}(Q)=1-y^2=1$ there. Hence
\[ \mathrm{Var}(\hat\eta_S)=\frac{1}{N_Q D_i^2}+o(N_Q^{-1}), \]
which is \eqref{eq:deltavar}. The angular sensitivity of the inverse is therefore the damping product $D_i$, equal to
$(1-2p)^n$ for the complete family at uniform rate, the first branch of $c_B$. Away from calibration the
leading variance is
$\bigl(z^2\mathrm{Var}(\hat y)+y^2\mathrm{Var}(\hat z)-2yz\,\mathrm{Cov}(\hat y,\hat z)\bigr)/(y^2+z^2)^2$,
which reduces to the displayed expression as $y\to0$. The second branch $\sqrt{\lambda_2}$ concerns the joint $(\etav,p)$ worst direction and enters when the stochastic block is
estimated alongside the angles.

\begin{lemma}[Probability-weight sandwich]\label{lem:sandwich}
Let $p_x$ run over the outcome probabilities of the measured settings at the reference, all positive, with
$p_{\min}=\min_xp_x$ and $p_{\max}=\max_xp_x$. Then
\[ \tfrac{1}{p_{\max}}\,J_p^\top J_p\ \preceq\ \FZ\ \preceq\ \tfrac{1}{p_{\min}}\,J_p^\top J_p, \]
so
\[ \tfrac{1}{2^np_{\max}}\,\sigmin(\JZ)^2\ \le\ \lambda_{\min}(\FZ)\ \le\ \tfrac{1}{2^np_{\min}}\,\sigmin(\JZ)^2, \]
and the same bounds hold for the smallest nonzero eigenvalues, since all three matrices share the kernel
$\ker J_p$.
\end{lemma}
\begin{proof}
$\FZ=J_p^\top\mathrm{diag}(1/p_x)J_p$ with $p_{\min}\le p_x\le p_{\max}$ gives the operator sandwich;
$\JZ=WJ_p$ with $W^\top W=2^nI$ per setting gives $J_p^\top J_p=2^{-n}\JZ^\top\JZ$; and $A\preceq B$ implies
$\lambda_{\min}(A)\le\lambda_{\min}(B)$.
\end{proof}
The parity conditioning therefore controls the Fisher conditioning with explicit constants whenever the
outcome probabilities are positive. $\FZ$ and $J_p$ share the kernel $\ker J_p$, so the comparison is on the
estimable subspace and the matching quantity is $\sigma^{+}_{\min}$. With $N$ samples of each measured
setting, the worst \emph{estimable} direction obeys
$N\cdot\mathrm{MSE}\ge1/\lambda^{+}_{\min}(\FZ)\ge2^np_{\min}/\sigma^{+}_{\min}(\JZ)^2$. For a full-rank
design $\sigma^{+}_{\min}=\sigmin$ and the bound is the one quoted in the main text. Below coverage the
full-space $\sigmin$ vanishes at $\eta=0$ and carries no bound of its own. The divergence of \Cref{thm:p5}(5)
is read along a full-rank approach $\etav\to0$, where $\sigma^{+}_{\min}(\etav)=\sigmin(\etav)=O(\|\etav\|)$
and the outcome probabilities stay bounded away from zero. Both constants equal $1$ at a
uniform reference and are $\approx1$ near it, giving $N\cdot\mathrm{MSE}\ge\sigmin^{-2}$ in parity
coordinates up to those factors.

\emph{Coordinates.} $c_B$ and $\sigmin$ are computed on the parity Jacobian $J_Z$. The unweighted
probability Jacobian $J_p$ ($J_Z=WJ_p$, $W^\top W=2^nI$) has singular values smaller by $2^{n/2}$. Sample complexity is set by the statistical Fisher
$F_Z=J_p^\top\mathrm{diag}(1/p_x)J_p$, which \Cref{lem:sandwich} ties to the parity value with explicit constants. At a near-uniform reference the probability weighting cancels the Walsh factor and the
Fisher $\sigmin$ equals the parity $\sigmin$ up to the factors
$(2^np_{\max})^{-1/2},(2^np_{\min})^{-1/2}\approx1$. At
the peaked $|0^n\rangle$ reference of the sweep the two differ in absolute scale, but the parity $\sigmin$
still tracks the Fisher conditioning. Across the full-rank design
population the parity $\sigmin$ and the Fisher $\sigmin$ ($\sqrt{\lambda_{\min}(F_Z)}$) rank the designs
almost identically (Spearman $\rho=0.984$ at $n=3$, $0.993$ at $n=4$), and the Fisher $\sigmin$ predicts
finite-sample RMSE comparably ($\rho=-0.65$ vs $-0.63$ at $n=3$, $-0.85$ vs $-0.85$ at $n=4$).
The design ordering that $c_B$ and the parity $\sigmin$ induce thus closely matches the ordering under the
exact Fisher metric. The absolute $c_B$ values are reported in parity coordinates. In the exact density-matrix computation, $\sigmin$ is nearly flat in $\eta$ above coverage, and below it the sub-coverage design
$Y^n$ at $n=3$ has $\log\sigmin$ against $\log\eta$ of slope $0.984$ ($95\%$ CI $[0.973,0.994]$, residual RMS
$0.029$) over $\eta\in[2\times10^{-3},0.25]$ on $14$ points, matching the predicted $O(\|\etav\|)$. Designs
that cover no generator at all, such as $X^n$, have $\sigmin$ numerically zero across that range, so no slope is defined for them. The separating code has $\rH=0$ with
$\sigmin=0.72,0.65,0.58,0.53,0.47,0.42$ at $n=3,\dots,8$, $|\Bstar|=2,3,3,3,3,4$.

\Cref{tab:coords} reports the separating code's conditioning floor as $n$ grows. The parity $\sigmin$ tracks
the floor $c_B=(1-2p)^n$ and decays exponentially in $n$, with the unweighted probability-Jacobian $\sigmin$ (the parity
value divided by $2^{n/2}$) listed alongside.

\begin{table}[h]
\centering
\caption{Separating-code conditioning floor as a function of $n$ (complete family, uniform $p=0.05$). The parity $\sigmin$ tracks the floor
$c_B=(1-2p)^n$. The unweighted probability-Jacobian $\sigmin$ (parity divided by $2^{n/2}$) is listed for
reference and is not the sample-complexity object. The floor decays exponentially in $n$, a fixed-$n$ guarantee.}
\label{tab:coords}
\resizebox{\columnwidth}{!}{%
\begin{tabular}{lcccccc}
\toprule
$n$ & 3 & 4 & 5 & 6 & 7 & 8 \\
\midrule
parity $\sigmin$ ($\approx c_B=(1-2p)^n$) & $0.72$ & $0.65$ & $0.58$ & $0.53$ & $0.47$ & $0.42$ \\
unweighted prob.-Jacobian $\sigmin$ & $0.25$ & $0.16$ & $0.10$ & $0.066$ & $0.042$ & $0.026$ \\
\bottomrule
\end{tabular}}
\end{table}

\paragraph{Conditioning across the design population.} \Cref{fig:sweep} measures the association between
$\sigmin$ and finite-sample error. Among the sampled random product-Pauli designs that reach full rank at the
tested reference and matched budget, RMSE tracks $\sigmin$ as a monotone anti-correlation (Spearman $\rho=-0.63$, bootstrap $95\%$ CI $[-0.73,-0.50]$, permutation $p<10^{-4}$, at $n=3$ over $250$ designs; $\rho=-0.85$, CI $[-0.91,-0.77]$, $p<10^{-4}$, at $n=4$ over $120$), with the lowest-$\sigmin$ quintile incurring $3.1$--$4.1\times$ the RMSE of the
highest-$\sigmin$ quintile at equal budget and shots. Designs are drawn with replacement, half from the full product-Pauli pool and half from the $XY$ pool, so
the $250$ and $120$ counts are draws over $94$ and $115$ distinct designs. They are selected by their computed $\sigmin$ and $\rH$ alone, before any fit, and the budget is matched across designs at a total of $8192$ shots per design,
split equally across that design's settings, aggregated over $24$ independent seeds per design (the same budget as \Cref{fig:sweep}). The bias--variance decomposition of \Cref{tab:biasvar} uses the same $8192$-shot total over $40$ seeds, whose
mean-squared-error ratios carry $95\%$ bootstrap intervals $68\times\,[53,87]$, $33\times\,[24,44]$, $25\times\,[19,33]$ and $17\times\,[13,23]$ at $\theta_{\mathrm{coh}}=0.05,0.10,0.15,0.20$, each excluding
unity. Code and
data reproducing every finite-sample result are covered by the Data Availability statement.
\emph{Signed-recovery protocol (\Cref{fig:sign}).} Each trial draws $|\eta_S|$ uniformly on $[0.08,0.20]$
with independent signs, and rates $p_S$ uniformly on $[0.02,0.08]$. The passive estimator is bounded
nonlinear least squares on the $Z^3$ parities with every $p_S$ a free nuisance parameter, bounds
$|\eta_S|\le0.6$ and $p_S\in[0,0.4]$, and four restarts whose angle initializations carry independent random
signs at $0.1$\,rad, the lowest-cost solution retained. No sign information is supplied to it. The code
estimator is \eqref{eq:inverse} evaluated on the empirical parities of $Z^3\cup\Bstar$, with the shot budget
split across settings so that the two strategies use the same total. The sign error rate is the fraction of
the six generators whose recovered sign is wrong, averaged over $60$ trials. The RMSEs are means over trials of the per-trial root-mean-square error, and intervals resample trials.

\emph{Sweep protocol.} The reference is $\eta_S=0.10\,r_S$ with fixed generic multipliers
$r_S\in[0.6,1.4]$ drawn once and reused for every design, and $p_S=0.05$ uniform. Each design is fitted by
bounded nonlinear least squares on the measured parities, with the stochastic rates fitted jointly as
nuisance parameters, bounds $|\eta_S|\le0.6$ and $p_S\in[0,0.4]$, and two fixed initializations
($0.05$\,rad and $0.15$\,rad per angle, the lower-cost solution retained), all chosen independently of the
injected values. The reported RMSE compares recovered and injected angle \emph{magnitudes}, so the sweep measures conditioning. This rank filter is evaluated at the nonzero reference and is
therefore weaker than structural coverage. The population contains designs that reach rank through
higher-order coherent mixing, and those occupy its ill-conditioned end.
The greedy $D$-optimal comparator adds settings one at a time from the $XY$ candidate pool (all
$X/Y$ strings, the pool the greedy searches use), each maximizing the log pseudo-determinant of $J^\top J$
(the $D$-optimality objective) at the same $n=3$ reference ($\eta\!=\!0.1$, $p\!=\!0.05$), up to the matched
budget $L=\lceil\log_2(n+1)\rceil$. It is a single greedy pass with no restarts and reaches $\sigmin=0.097$
at $n=3$. The E-optimal comparison of \Cref{sec:exp} instead searches the \emph{full} product-Pauli pool
$\{X,Y,Z\}^n$ exhaustively over all $L$-subsets. At $n=3$ the best conditioning is $0.71954$
($\{YXY,YYX\}$) against the code's $0.71703$, and at $n=4$ it is $0.82293$ ($\{XYYZ,YYZX,ZYXY\}$) against the
code's $0.64249$. Of the full-rank designs in the pool, $1$ of $206$ beats the code at $n=3$ and $805$ of
$85\,304$ at $n=4$. The code is a closed-form construction with a proved floor, but at $n=4$ it is not the
$\sigmin$-maximizer.

\begin{remark}[Explicit radius and SPAM robustness]\label{rem:t5}
$\|J_M(\etav,p)-J_M(0,p)\|_2\le L_M\|\etav\|_2$ with $L_M=m\sqrt{5R_M}$ ($m=n(n+1)/2$, $R_M=(|B|+1)2^n$),
so $\sigmin(\etav)\ge c_B-L_M\|\etav\|$ and the explicit radius is $\eta_0=c_B/(2L_M)$. Independent symmetric readout error at rate $\varepsilon_r$ on each qubit scales row $P$ by $(1-2\varepsilon_r)^{w(P)}$, so
$\sigmin(J_r)\ge(1-2\varepsilon_r)^{w_{\max}}\sigmin(J)$. An input mixture
$\rho_0\to(1-\varepsilon_s)\rho_0+\varepsilon_s\rho_{\mathrm{err}}$ gives
$\sigmin\ge(1-\varepsilon_s)\sigmin(J)-\varepsilon_s\|J_{\mathrm{err}}\|_2$, with $J_{\mathrm{err}}$ the same
Jacobian evaluated at $\rho_{\mathrm{err}}$. These constants are conservative but explicit. At
$p=0.05$ and $n=3$ the guaranteed radius is $\eta_0\approx0.0055$, below the smallest angle we inject
($\approx0.025$), since $\eta_0$ is a worst-case Lipschitz guarantee. Over the tested range the observed
degradation is far milder.
\end{remark}

% ======================================================================
\section{Proof of \texorpdfstring{\Cref{thm:separation}}{Theorem 5}, the conjugation reduction, and input invariance}\label{app:separation}

The $X$-type model is no loss of generality for the transverse class, and non-identifiability is a property of the input state and data map, with the canonical $|0^n\rangle$ as one instance.

\paragraph{Conjugation reduction (used in \Cref{sec:setup}).} A pairwise-commuting transverse family (each
generator a Pauli $P_g$ with $X$-support $a_g$ and $Y$-set $b_g\subseteq a_g$) of distinct nonempty
weight-1 and weight-2 supports admits a global gauge $s\in\mathbb{F}_2^n$ with $b_g=a_g\wedge s$, because distinct weight-1 and weight-2 supports overlap in $\le1$ qubit and symplectic commutation on the overlap forces the
$Y$-pattern to agree. The local Clifford $C_s=\prod_{s_i=1}S_i$ satisfies $C_sP_gC_s^\dagger=\varepsilon_g
X_{a_g}$ with $\varepsilon_g=(-1)^{|b_g|}$, fixes every $Z$-diagonal input and $Z^n$, and maps product
settings by $X\!\leftrightarrow\!Y$ on the $s$-qubits. The data maps satisfy $J_{\mathrm{tr}}=\Pi J_X D$ and
$F_{\mathrm{tr}}=D^\top F_X D$ with $D=\mathrm{diag}(\varepsilon,1)$, so rank, kernel dimension, singular
values and $c_B$ are identical, and the null-space basis is $v_{\mathrm{tr}}(u)=(\varepsilon\odot u,\tfrac12
d\odot c(u))$. Hence all theorems transfer verbatim with $\Bstar_s:=\sigma_s(\Bstar)$, justifying the
$X$-type WLOG. (Weight $3$ admits non-gaugeable commuting families, e.g.\ $\{Y_1X_2X_3,X_1Y_2\}$, and the design theory above is stated for weights one and two.)

\paragraph{$Z$ readout at zero angle.} For an arbitrary Pauli family on a $Z$-diagonal input, at $\eta=0$ all
coherent columns of the $Z$-histogram Jacobian vanish (off-diagonal insertion never returns to the diagonal),
and the stochastic block equals the same-support $X$-type block (flips act as abelian XOR mixtures), so
$\rH(Z)|_{\eta=0}=m+\mathrm{corank}\,J_{\mathrm{stoch}}$ universally.

\begin{proof}[Proof of \Cref{thm:separation}]
For a weight-1 and weight-2 transverse class under Pauli tomography, identifiability near calibration is equivalent to the
$X$-supports being nonzero and pairwise distinct. Equal or empty supports collide in the incidence matrix
$B$ and leave a stochastic-block null direction, while distinct nonzero supports make $B$ full column rank,
and the coverage construction of \Cref{app:separatrix} supplies the coherent block. A pure longitudinal
$Z$-channel fixes every $Z$-diagonal state for all $(\etav,p)$, so every column of the $|0^n\rangle$ Jacobian
vanishes and no measurement (product, entangled, or adaptive) identifies it. Conjugating by $H^{\otimes n}$
(equivalently, the input $|{+}^n\rangle$ with the readout rotated by $H^{\otimes n}$) maps it to the transverse theory, which is identifiable. Hence the
obstruction is the input.
\end{proof}

\paragraph{Input invariance (used in \Cref{sec:input,sec:exp}).} Writing $M(\theta)$ for the Pauli-transfer
matrix, each Jacobian entry $[\JZ]_{T,k}=\sum_P[\partial_{\theta_k}M(\theta)]_{Z_T,P}\,r(\rho_0)_P$ is linear in $r(\rho_0)$ with trigonometric-polynomial coefficients. The rank is therefore constant
off a proper algebraic locus, so there is a generic value $\rH^\star(n,K)$ attained for
Zariski-almost-every input. The genericity argument proves that the rank is constant, and the computation supplies its value. For
$K\in\{1,2\}$ we find by exact computation through $n=8$ that the generic rank saturates the outcome bound,
$\rH^\star(n,K)=\max(0,\,2|G(n,K)|-(2^n-1))$, with $\rH^\star(3,2)=5$. At $K=1$ the canonical input is the exception, with $\rH=|G|$ by \Cref{lem:n3}. Whether the bound is saturated for all $n$ remains open.

% ======================================================================
\section{Twirling and the recovered coherent information}\label{app:twirl}

Twirling removes the coherent parameter as an identifiable degree of freedom, whereas the separating code recovers it.

\begin{proposition}[Recovery versus suppression]\label{prop:twirl}
Pauli twirling each generator channel on its support maps $(\eta_S,p_S)$ to a single effective rate
$q_S=p_S+\sin^2(\eta_S/2)-2p_S\sin^2(\eta_S/2)$, so the twirled channel is
$\prod_S\mathrm{Flip}_{S,q_S}$ and for any measurement $M$,
$\mathrm{rank}\,J_M(\theta)\le|G|$, $r_H^{\mathrm{twirl}}(M)\ge|G|$, so the independent coherent parameter is
not identifiable after twirling. Without twirling, the code of \Cref{thm:separatrix} covers, so $\rH=0$ at $\eta=0$
(\Cref{prop:jointrank}) and the signed angles are recovered by \eqref{eq:inverse} on the box.
\end{proposition}
\begin{proof}
The twirl projects $\mathrm{Ad}_{R_S}$ onto its Pauli-diagonal part, deleting the map
$P\mapsto iX_SP$ and keeping the multiplier $\cos\eta$ on anticommuting $P$. This equals the Pauli channel
with rate $\sin^2(\eta_S/2)$, which composes with $\mathrm{Flip}_{S,p_S}$ by XOR to give $q_S$. The outcome
distribution depends on $\theta$ only through $q(\theta)$, so $J_M=J_{M,q}\,Dq$ with $Dq\in\mathbb{R}^{|G|
\times2|G|}$, bounding the rank; the per-generator direction $v_S=e_{\eta_S}-\tfrac12 d_S\tan\eta_S\,e_{p_S}$
satisfies $Dq\,v_S=0$. The separating-code recovery claim is \Cref{thm:separatrix}(ii).
\end{proof}
After twirling, $\eta_S$ enters the channel only through the effective rate $q_S$, so it survives but is no longer an independent identifiable degree of freedom. Twirling is a
suppression route (randomized compiling~\cite{wallman2016noise}, output a Pauli channel), while the separating
code is a recovery route
(deterministic basis changes, the coherent channel kept).

% ======================================================================
\section{Hardware consistency check: protocol and per-replicate results}\label{app:hardware}

\paragraph{Protocol.} We inject known coherent over-rotations and stochastic flips on $n=3$ qubits and
compare two estimators of the injected angles: a single fixed-basis ($Z$) least-squares fit, and the
separating-code fit using $Z^3$ plus the $\lceil\log_2 4\rceil=2$ settings of \Cref{thm:separatrix}. Two
IBM Heron superconducting processors, labeled A and B, were used, each on a linear chain of four qubits,
three carrying the $n=3$ register and the fourth used for routing.
The chains passed a \emph{recovery-stability} prescreen applied before the estimator comparison and evaluated
on an $8$-circuit probe ($Z^3$ and the two code settings at two injected-angle scales, with two calibration
circuits), and the separating-code fit on those circuits had to satisfy $|\mathrm{intercept}|\le0.02$\,rad and
slope in $[0.85,1.15]$. The criterion is therefore evaluated on the code's own recovery, and the retained
chains are ones on which that recovery is stable. The same probe fitted with $Z^3$ alone gives intercepts
of $0.12$--$0.17$\,rad and slopes of $0.25$--$0.63$, so no chain would pass a single-basis version of it.
Two earlier backends were screened out for fit drift before this comparison. On the two
retained backends we obtained three replicates on Heron A and two on Heron B, each sweeping four injected-angle scales ($\eta_{\mathrm{scale}}\in\{0.05,0.1,0.15,
0.2\}$ with fixed generic per-generator multipliers: $0.9170$, $0.8829$, $1.0973$ on the three
single-qubit generators and $0.7049$, $1.2311$, $0.7276$ on the three two-qubit generators) at $8192$ shots
per circuit.

\paragraph{Results.} The \emph{magnitude error} of a fit is the median over the six coherent generators of
$\lvert|\hat\eta_S|-|\eta^{\mathrm{inj}}_S|\rvert$, and the \emph{error ratio} at a given scale is the
single-basis magnitude error divided by the separating-code one. The comparison is on magnitudes, so it tests the conditioning ordering. Both fits use one fixed initialization ($0.05$\,rad per coherent angle, $0.02$ per stochastic rate), chosen
independently of the injected values. At the smallest scale that initialization is comparable in magnitude to
the injected angles, so the per-replicate value is the median of the ratio over the three larger scales
$\eta_{\mathrm{scale}}\in\{0.1,0.15,0.2\}$. Excluding the smallest scale does not inflate the result, since the pooled median over all four scales is higher. The separating-code magnitude error (the denominator, the
``Code err'' column of \Cref{tab:hw}) stays in $0.013$--$0.018$\,rad, bounded away from zero, so the single-basis
value (the numerator) is correspondingly $\approx0.04$--$0.09$\,rad. The gap also persists at a matched total budget of $\approx8192$ shots per strategy
(the code's three settings subsampled to $2730$ each, $8190$ in all, against the single basis' $8192$ on $Z^3$), and the single fixed basis stays $2.8$--$4.0\times$ worse (median $3.5\times$, over $15$ subsample
draws), so the code's larger shot count does not explain the gap. \Cref{fig:hw} reports, per replicate: \emph{(a)} the ratio of the single-basis fit's
magnitude error to the separating-code fit's, which spans $2.97$--$5.05$ (median $3.96$), so the single basis carries
$3$--$5\times$ the error; and \emph{(b)} the separating-code fit's injected-vs-recovered-angle slope, which spans $0.88$--$1.05$ (near-unit recovery). The code's magnitude error stays $\le0.019$\,rad on all replicates. On every replicate the single basis carries the larger magnitude error, in agreement with the computed conditioning, while the size of the ratio varies with backend and replicate.

\begin{table}[h]
\caption{Per-replicate hardware results (five eligible replicates on two prescreened IBM Heron processors
A and B, $n=3$, injected known errors). Prescreen intercept (rad) and slope are the separating-code
recovery-stability prescreen; ``error ratio'' is the single-basis fit's magnitude error divided by the separating-code fit's;
``code slope'' and ``code err'' are the separating-code fit's injected-vs-recovered slope and magnitude error.}
\label{tab:hw}
\centering
\resizebox{\columnwidth}{!}{%
\begin{tabular}{llcccc}
\hline
Backend & Rep & Prescreen (ic, slope) & Error ratio & Code slope & Code err (rad)\\
\hline
A & 1 & $0.016$, $0.99$ & $2.97$ & $0.99$ & $0.018$\\
A & 2 & $0.005$, $1.07$ & $5.05$ & $1.05$ & $0.018$\\
A & 3 & $0.011$, $0.98$ & $4.57$ & $0.98$ & $0.013$\\
B & 1 & $0.012$, $0.96$ & $3.96$ & $0.95$ & $0.017$\\
B & 2 & $0.013$, $0.90$ & $3.41$ & $0.88$ & $0.015$\\
\hline
\end{tabular}}
\end{table}

The single-basis fit carries more magnitude error than the code on every one of the five prescreened replicates
(\Cref{tab:hw}). The ratio also exceeded one on every excluded (drifting) chain (raw single-basis-to-code
error ratios spanning $1.8$--$7.9$), so the ordering holds on the excluded chains as well.

\begin{figure*}[t]
\centering
\includegraphics[width=\textwidth]{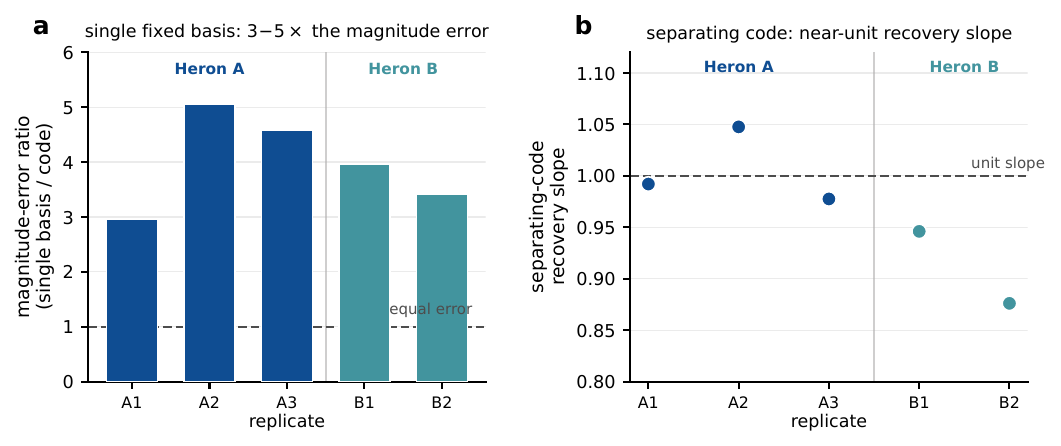}
\caption{Hardware consistency check on two IBM Heron processors ($n=3$, injected known errors, five eligible
replicates; Heron A: $3$ reps, Heron B: $2$ reps). Magnitude error is the median over the six coherent
generators of $\lvert|\hat\eta_S|-|\eta^{\mathrm{inj}}_S|\rvert$, and each replicate's ratio is the median over
the three larger injected-angle scales. \textbf{(a)} The single fixed-basis fit carries $3$--$5\times$ the
magnitude error of the separating-code fit (ratio $>1$ above the equal-error line; range $2.97$--$5.05$).
\textbf{(b)} The separating-code fit recovers the injected angles with near-unit slope (range
$0.88$--$1.05$).}
\label{fig:hw}
\end{figure*}

% ======================================================================
\section{Robustness to support misspecification}\label{app:misspec}
The exact theory assumes the generator support $G$ is known. We evaluate sensitivity to that assumption on the
separating code $\Bstar\cup Z^n$ at $n=3$ (reference $\eta=0.1$, $p=0.05$, a matched $8192$ shots per design
split equally across settings, $24$ seeds), comparing a correctly specified support against two one-generator
perturbations: \emph{missing} a true weight-$2$ generator (the noise contains it, the fit omits it) and a
\emph{spurious} generator (the noise lacks it, the fit includes it). Angle-recovery RMSE is measured on the
generators common to truth and fit, and each misspecified fit is compared against a correct-support fit scored
over those \emph{same} common generators under the same data-generating truth (a matched baseline). Relative to
that baseline, the common-angle RMSE is $1.01\times$ when a true generator is missing and $1.00\times$ when a
spurious one is added. In the spurious case the extra generator takes a fitted angle of $0.012$\,rad. Recovery of the remaining angles is thus insensitive to this single
weight-$2$ support perturbation at the tested $(n{=}3,\eta,p)$ reference.

\section{Non-quantum instances}\label{app:ml}

Two classical statistical-learning settings exhibit the same singular-Fisher obstruction and the same
qualitative rank and conditioning behavior.

\begin{remark}[Two classical instances of the singular-Fisher obstruction]\label{rem:ml}
\emph{(i) Controlled active learning.} Logistic regression on $\mathbb{R}^d$ whose query pool spans only
a proper subspace $V\subsetneq\mathbb{R}^d$. A weight direction $w_\perp\perp V$ that the pool never excites
lies exactly in the kernel of the averaged Fisher information, because two weights $w$ and $w+t\,w_\perp$ induce identical label distributions on the pool for all $t$. The component along $w_\perp$ therefore stays unresolved at every sample size $N$.
Designed queries spanning the complement of $V$ restore identifiability, with error $\sim1/\sqrt N$. At equal
measurement energy, a covering but ill-conditioned query design incurs an order of magnitude more finite-sample
cost than the well-conditioned one.

\emph{(ii) Overparameterized random features (\Cref{fig:c7}).} A frozen random-feature representation
$\phi(x)=\tanh(xW+b)$ of real images (the $8\times8$ digits), dimension $D=150$, with an overparameterized
labeled pool of $n_0=60<D$ examples and label noise $0.3$, fit by ridge regression ($\lambda=0.1$). Because
$n_0<D$, the pool Fisher information $\Phi_{\mathrm{pool}}^\top\Phi_{\mathrm{pool}}/n_0$ has a null space of dimension at least $D-n_0=90$ (the elementary underdetermination bound). The pool likelihood is flat along that null, so repeated labels on the same pool leave the predictor's component there unidentified. Ridge
($\lambda=0.1$) resolves the indeterminacy by setting that component to zero, a generically wrong guess that
surfaces as irreducible held-out test error (excess MSE relative to the noiseless teacher). The cure is active acquisition of \emph{real} held-out examples. Among three strategies at equal
budget, \emph{conditioned} (acquire examples with the largest null-space projection $\|N^\top\phi\|$),
\emph{random}, and \emph{smallest projection}, the largest-projection
strategy's error stays flat at the first step ($0.240\!\to\!0.243$) then falls toward $0.02$, whereas random
and smallest-projection acquisition \emph{raise} the error sharply ($0.32$, $0.55$) before recovering. Largest projection
$<$ random $<$ smallest projection at every \emph{positive} budget (the three tie at the pool-only point
$0.240$). The
rank deficiency arises from the elementary $n_0<D$ dimension count alone. Reported over $30$ seeds with $95\%$ bootstrap intervals.
\end{remark}

\begin{figure}[t]
\centering
\includegraphics[width=\linewidth]{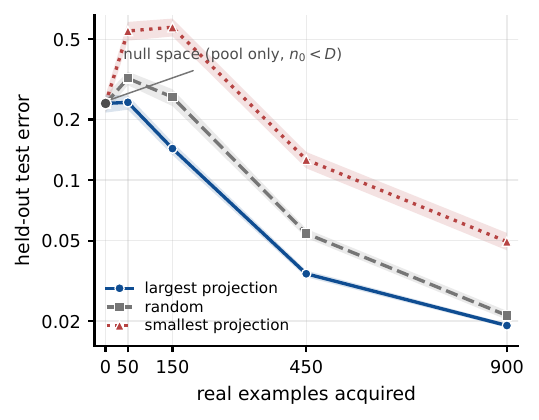}
\caption{A classical analogue of the obstruction in overparameterized random features ($D=150$, pool
$n_0=60<D$, label noise $0.3$, ridge $\lambda=0.1$). The pool Fisher information has a null space of dimension at least $D-n_0=90$, so repeated labels on the same pool leave the predictor's component there unidentified.
Active acquisition of real held-out examples selected for the largest null-space projection stays flat at the
first step ($0.240\!\to\!0.243$) then falls toward $0.02$, while random and smallest-projection acquisition raise
the error sharply ($0.32$, $0.55$) before recovering. Reported over $30$ seeds with $95\%$ bootstrap
intervals.}
\label{fig:c7}
\end{figure}